\documentclass{article}

\usepackage[margin=1in]{geometry}
\usepackage{amsmath,amsfonts,amssymb,amsthm,mathtools}
\usepackage{hyperref}

\numberwithin{equation}{section}
\newtheorem{thm}{Theorem}[section]
\newtheorem{lem}[thm]{Lemma}
\newtheorem{prop}[thm]{Proposition}

\theoremstyle{definition}
\newtheorem*{def*}{Definition}
\newtheorem{defn}[thm]{Definition}
\newtheorem*{rem*}{Remark}
\newtheorem{rem}[thm]{Remark}

\DeclareMathOperator{\Res}{Res}
\DeclareMathOperator{\re}{Re}

\DeclareMathOperator{\ad}{ad}
\newcommand{\IP}[2]{\langle #1, #2 \rangle}

\allowdisplaybreaks

\usepackage{authblk}
\title{Baker--Akhiezer function for the deformed root system $BC(l,1)$ and bispectrality} 
\author{Iain McWhinnie, Liam Rooke, Martin Vrabec\footnote{Corresponding author, \texttt{martinvrabec222@gmail.com}}}
\date{\today}
\affil{School of Mathematics and Statistics, University of Glasgow, University Place, Glasgow G12 8QQ, UK}

\AtEndDocument{\bigskip{\footnotesize%
  \textit{E-mail addresses:} \, \texttt{iainmcwhinniemaths@gmail.com, } 
  \texttt{liamrooke01@gmail.com, }
  \texttt{martinvrabec222@gmail.com} \par
}}

\date{}

\begin{document}
\maketitle

\begin{abstract}
    We show that a Sergeev--Veselov difference operator of rational Macdonald--Ruijsenaars (MR) type for the deformed root system $BC(l,1)$ preserves a ring of quasi-invariants in the case of non-negative integer values of the multiplicity parameters. We prove that in this case the operator admits a (multidimensional) Baker--Akhiezer eigenfunction, which depends on spectral parameters and which is, moreover, as a function of the spectral variables an eigenfunction for the (trigonometric) generalised Calogero--Moser--Sutherland (CMS) Hamiltonian for~$BC(l,1)$. 
    By an analytic continuation argument, we generalise this eigenfunction also to the case of more general complex values of the multiplicities. This leads to a bispectral duality statement for the corresponding MR and CMS systems of type~$BC(l,1)$. 
\end{abstract}

\section{Introduction}
The Calogero--Moser--Sutherland (CMS) models are an important example of integrable many-body Hamiltonian systems in one spatial dimension. Their study goes back to the works of Calogero~\cite{C'71}, Sutherland~\cite{S'72}, and Moser~\cite{M'75}, who investigated certain problems of pairwise interacting particles on a line or a circle.
Olshanetsky and Perelomov observed a connection between the original CMS Hamiltonians and the root system~$A_l$ ($l \in \mathbb{Z}_{>0}$), and they generalised these Hamiltonians to the case of arbitrary root systems of Weyl groups in~\cite{OP'76, OP'77} in a way that preserved integrability. 

Chalykh, Feigin, and Veselov showed in~\cite{VFC, CFV'98, CFV'99} in the quantum case that there exist further integrable generalisations related to other special configurations of vectors that are not root systems. The examples they discovered were certain one-parametric deformations $A(l,1)$ and $C(l,1)$ of the root systems~$A_{l+1}$ and~$C_{l+1}$, respectively. Other examples have been discovered since, one of which is the deformation~$BC(l,l')$ of the root system $BC_{l+l'}$ that was first considered in~\cite{Sergeev_2004} by Sergeev and Veselov. An elliptic version of the special case~$BC(l,1)$ appeared earlier in~\cite{CEO}.

In general, the eigenfunctions of quantum integrable systems may be complicated. However, the generalised CMS operators associated with any root system with Weyl-invariant integer values of multiplicity parameters admit as (singular) eigenfunctions so-called multidimensional Baker--Akhiezer (BA) functions, which are relatively elementary functions~\cite{Chalykh_1990, Chalykh_1993, CFV'99}. They have the form 
\begin{equation*}
    \psi(z,x) = P(z,x)e^{\IP{z}{x}},
\end{equation*}
where $P$ is a polynomial in a set of complex variables~$z$, with $P$ depending also on another set of complex variables~$x$ that are the variables in which the CMS operators act, and $\IP{\cdot}{\cdot}$ denotes the standard real Euclidean inner product extended $\mathbb{C}$-bilinearly. 

The function~$\psi$ can be characterised by its properties as a function of the variables~$z$. 
Such an axiomatic definition of a multidimensional BA function was proposed by Chalykh, Styrkas, and
Veselov for an arbitrary finite collection of non-collinear vectors with
integer multiplicities in~\cite{Chalykh_1993} --- the case of (positive subsystems of) reduced root systems was considered earlier in~\cite{Chalykh_1990} --- and see also~\cite{Feigin_2005} for a weaker version of the axiomatics. For the (only) non-reduced root system~$BC_l$, an axiomatic definition of the BA function was given by Chalykh in~\cite{Chalykh_2000}. 
Recently, Feigin and one of the authors extended the definition to the case of configurations where collinear vectors are allowed more generally as long as all subsets of collinear vectors are of the form $\{\alpha$, $2\alpha\}$~\cite{Feigin_2022}.
The key properties that the function needs to satisfy are
quasi-invariance conditions of the form
\begin{equation*}
    \psi(z + s\alpha, x) = \psi(z - s\alpha, x)
\end{equation*}
at $\IP{\alpha}{z} = 0$ for vectors $\alpha$ in the configuration, where $s$ takes
special integer values depending on the multiplicities.

Such a BA function can exist only for very special configurations. The corresponding generalised CMS operators have to be algebraically integrable, that is, contained in a large commutative ring of differential operators; and if the BA function exists, then it is a common eigenfunction for all the operators in this ring~\cite{Chalykh_1993, Feigin_2005, Feigin_2022}. 
In addition to the root systems case, BA functions have been found also for some other configurations. The known
examples include the aforementioned deformations from~\cite{CFV'98, CFV'99} of the root systems of type $A$ and $C$ in the case when the multiplicity parameters are integer. The corresponding BA functions were constructed in~\cite{Chalykh_2000, Feigin_2005}, respectively. Another deformation~$A_{l,2}$ of the root system of type~$A$ appeared in~\cite{CVlocus}, and the BA function (satisfying the weakened axiomatics) for it was given in~\cite{Feigin_2005}. More recently, a new configuration, called~$AG_2$, emerged in the work of Fairley and Feigin~\cite{FF}, for which the BA function was constructed in~\cite{Feigin_2022}. The corresponding generalised CMS operator was studied in~\cite{Feigin_2019}.


The deformed CMS operator associated with the configuration $BC(l,l')$ from~\cite{Sergeev_2004} is algebraically integrable in the case when $l'=1$ and all the multiplicity parameters are integer~\cite{Chalykh_2008}, which suggested that in this case there might exist a BA function for it. This has been the only remaining known example of an algebraically integrable monodromy-free CMS-type operator for which a BA function had not been written down.  
In this paper, we explicitly construct it by a method modelled on Chalykh's one from the paper~\cite{Chalykh_2000} (see~\cite{Feigin_2005, Feigin_2022} for further examples where such a technique has been applied), utilising also some general results proved in~\cite{Feigin_2022}. The construction uses an integrable difference operator, acting in the variables~$z$, of rational Macdonald--Ruijsenaars type for the configuration~$BC(l,1)$ that was introduced by Sergeev and Veselov in~\cite{Sergeev_2009}.

The Macdonald--Ruijsenaars (MR) difference operators may be viewed as difference versions of the CMS Hamiltonians and their quantum integrals. They were introduced by Ruijsenaars in~\cite{R'87} for the root system of type~$A$, and in the case of arbitrary reduced
root systems by Macdonald~\cite{M}, and for the root system $BC_l$  by Koornwinder~\cite{Koornwinder_1992}. 
Their rational versions satisfy a bispectral duality relation with the (trigonometric) CMS Hamiltonians associated with the same root system. This duality was conjectured by Ruijsenaars~\cite{R'90}, and proved by Chalykh in~\cite{Chalykh_2000}. Namely, there exists a function $\Psi$ of two sets of variables $x$ and $z$ such that
\begin{equation}\label{eqn: bispectrality}
    H \Psi = \lambda \Psi, \quad D \Psi = \mu \Psi,
\end{equation}
where $H= H(x, \partial_x)$ is the generalised CMS Hamiltonian or its integral, $D$ is a rational MR operator acting in the spectral variables~$z$, and $\lambda = \lambda(z)$, $\mu = \mu(x)$ are the respective eigenvalues.
The relations~\eqref{eqn: bispectrality} are a higher-dimensional differential-difference analogue of the one-dimensional differential-differential bispectrality relation studied by Duistermaat and Gr\"unbaum~\cite{DG}. 

In the case of integer (and Weyl-invariant) multiplicities, the function~$\Psi$ in~\eqref{eqn: bispectrality} can be taken to be the BA function for the root system in question --- and the case of non-integer multiplicities can be handled by an analytic continuation argument~\cite{Chalykh_2000}. Moreover, Chalykh showed that the operator~$D$ can be used to explicitly construct the BA function itself. The key property needed for this from the operator $D$ is the preservation of a space of quasi-invariant analytic functions.

We note that the relation~\eqref{eqn: bispectrality} for the root system~$A_l$
with multiplicity parameter greater than~$1$ and for a function~$\Psi$
given by a Mellin--Barnes type integral was obtained by Kharchev and Khoroshkin in~\cite{KK}.

A form of bispectrality may also be seen in terms of special families of multivariable orthogonal polynomials (rather than a single function $\Psi$ depending on spectral parameters). In the case of the root system $A_l$, these are the Jack polynomials, and for other root systems they are the multivariable Jacobi polynomials, which admit Pieri-type formulas that can be interpreted as bispectrality between the CMS Hamiltonians 
and difference operators acting on the weights indexing the polynomials~\cite{McD, HO, Chalykh_2000}. 

Let us also mention that a version of the notion of a BA function exists for the rational degeneration of the CMS Hamiltonians $H$ as well. When suitably normalised, this function is symmetric under the exchange of~$x$ and~$z$, and the rational version of $H$ is bispectrally self-dual~\cite{Chalykh_1993}. 
BA functions for the trigonometric MR operators and their bispectrality properties were investigated in~\cite{Chalykh_2002}.

The paper~\cite{Chalykh_2000} proved also an analogue of the duality~\eqref{eqn: bispectrality} for the deformed root system~$A(l,1)$ and corresponding generalisation~$D$ of the rational limit of Ruijsenaars' operators. For the deformed root systems~$C(l,1)$ and $A_{l,2}$, this was done
(in the case of integer multiplicity parameters) in~\cite{Feigin_2005}, and the case of $AG_2$ was treated in~\cite{Feigin_2022}.
In this paper, we prove an analogue of the duality~\eqref{eqn: bispectrality} for the deformed root system~$BC(l,1)$, where $D$ will be Sergeev's and Veselov's rational difference operator for~$BC(l,1)$~\cite{Sergeev_2009}, and~$H$ their $BC(l,1)$-type generalised CMS Hamiltonian~\cite{Sergeev_2004}. In particular, for special integer values of the multiplicities, one recovers the results from~\cite{Feigin_2005} for the configuration~$C(l,1)$. 
Another bispectrality property of the Sergeev--Veselov operators for $BC(l,1)$ (as well as for $BC(l,l')$) in terms of super-Jacobi polynomials was proved in~\cite{Sergeev_2009}.

The structure of the paper is as follows. In Section~\ref{sec: BC(l,1)}, we describe the configuration~$BC(l,1)$ and its associated deformed CMS Hamiltonian. In Section~\ref{sec: difference operator}, we recall the Sergeev--Veselov difference operator for~$BC(l,1)$ and prove that, when the multiplicity parameters are non-negative integers, it preserves a ring of quasi-invariants. We use this property in Section~\ref{sec: BA function} to construct the BA function for $BC(l,1)$ and prove that it satisfies bispectrality. We generalise that bispectral duality statement to more general complex values of the multiplicities in Section~\ref{sec: non-integer multiplicities}.


\section{Deformed root system $BC(l,1)$} \label{sec: BC(l,1)}

The root system $BC_{l+1}$ has
a positive half $$BC_{l+1,+}=\{ e_i, \, 2e_i \,|\, 1 \leq i \leq l+1 \} \cup \{ e_i \pm e_j \,|\, 1 \leq i < j \leq l+1 \},$$
where $e_i$ denote the standard orthogonal unit vectors in $\mathbb{R}^{l+1}$.  
In this paper, we consider a deformation~$BC(l,1)$ of the root system~$BC_{l+1}$ with a positive half   
\[
BC(l,1)_{+}= \{ e_i, \, 2e_i,\, e_i \pm \sqrt{k}e_{l+1} \,|\, 1 \leq i \leq l \} \cup \{\sqrt{k} e_{l+1},\, 2\sqrt{k} e_{l+1} \} \cup \{e_i \pm e_j \, | \, 1 \leq i < j \leq l\} \subset \mathbb{C}^{l+1},
\]
where $k$ is a non-zero complex parameter~\cite{CEO, Sergeev_2004}.
Let $BC(l,1)^r$ denote the reduced version of this system whose positive half is
\begin{align*}
    BC(l,1)^r_+ &= \{ \alpha \in BC(l,1)_{+} \, | \, \tfrac12 \alpha \notin BC(l,1)\} \\
    &= \{ e_i, \, e_i \pm \sqrt{k}e_{l+1} \,|\, 1 \leq i \leq l \} \cup \{\sqrt{k} e_{l+1}\} \cup \{e_i \pm e_j \, | \, 1 \leq i < j \leq l\}.
\end{align*}
The vectors in the set $BC(l,1)_+$ are assigned certain multiplicities, which we specify below.

A multiplicity map on a set $\mathcal{A} \subset \mathbb{C}^{l+1}$ is a function $\mathcal{A} \to \mathbb{C}$ which assigns to each vector $\alpha \in \mathcal{A}$ a complex number $m_\alpha$ called its multiplicity. In Sections~\ref{sec: difference operator} and~\ref{sec: BA function}, we will be specifically interested in the case when the multiplicities are contained in~$\mathbb{Z}_{\geq 0}$.
Then for any $\alpha \in \mathcal{A}$, we define the set $$A_\alpha \coloneqq \{1,2,\dots,m_\alpha\} \cup \{ m_\alpha +2, m_\alpha +4, \dots, m_\alpha + 2m_{2\alpha}\},$$ where $m_{2\alpha}\coloneqq 0$ if $2\alpha \notin \mathcal{A}$.

The set $BC(l,1)_+$ has its multiplicity map given by $m_{e_i}=m$, $m_{2e_i}=n$, $m_{e_i \pm \sqrt{k}e_{l+1}}=1$, $m_{e_i \pm e_j} = k$, $m_{\sqrt{k}e_{l+1}}=p$, and $m_{2\sqrt{k}e_{l+1}}=r$ for complex parameters $m,n,p,r$, subject to the constraint that~$m=kp$ and~$2n+1 = k(2r+1)$. For $m=p=0$, the configuration $BC(l,1)$ reduces to the configuration $C(l,1)$, which was considered in~\cite{Feigin_2005}.
For $k=1$, the configuration $BC(l,1)$ reduces to the root system $BC_{l+1}$ with a Weyl-invariant assignment of multiplicities such that the vectors $e_i \pm e_j$ for $1 \leq i < j \leq l+1$ have multiplicity~$1$. In~Sections~\ref{sec: difference operator} and~\ref{sec: BA function}, we assume that $m,n,p,r \in \mathbb{Z}_{\geq 0}$, 
and if $l>1$ then also that $k \in \mathbb{Z}_{>0}$.

The generalised CMS operator associated with a finite collection of vectors $\mathcal{A} \subset \mathbb{C}^{l+1}$ with prescribed multiplicities has the form 
\begin{equation}\label{eqn: BC(l,1) CMS Hamiltonian}
    H = -\Delta + \sum_{\alpha \in \mathcal{A}} \frac{m_\alpha(m_\alpha + 2 m_{2\alpha} + 1)\langle \alpha, \alpha \rangle}{\sinh^2 \langle \alpha, x \rangle},
\end{equation}
where $x = (x_1, \dots, x_{l+1}) \in \mathbb{C}^{l+1}$ and $\Delta = \sum_{i=1}^{l+1} \partial_{x_i}^2$ ($\partial_{x_i} = \partial/\partial x_i$) is the Laplace operator on $\mathbb{C}^{l+1}$.
The Olshanetsky--Perelomov operators correspond to letting $\mathcal{A}$ be a positive half of a root system with a Weyl-invariant assignment of multiplicities. 
If one puts $\mathcal{A} = BC(l,1)_+$ in the formula~\eqref{eqn: BC(l,1) CMS Hamiltonian}, one obtains the generalised CMS operator associated with the configuration~$BC(l,1)$~\cite{Sergeev_2004}.

\section{Sergeev--Veselov difference operator for $BC(l,1)$}\label{sec: difference operator}
In this section, we recall the rational difference operator introduced by Sergeev and Veselov for the deformed root system~$BC(l,1)$~\cite{Sergeev_2009}, which deforms the rational version of Koornwinder's operator and also generalises an operator associated with $C(l,1)$ from~\cite{Feigin_2005}.
We prove that, when all the multiplicity parameters are non-negative integers, the operator preserves a ring of quasi-invariants $\mathcal{R}_{BC(l,1)}$ defined as follows.
    Let $\mathcal{R}_{BC(l,1)}$ be the ring of all analytic functions $f \colon \mathbb{C}^{l+1} \to \mathbb{C}$ such that 
    \begin{equation}\label{quasi}
    f(z+s\alpha) = f(z-s\alpha) \, \, \mathrm{at} \, \langle{z,\alpha}\rangle = 0
    \end{equation}
    for all $\alpha \in BC(l,1)^r_+$ and $s \in A_\alpha$.

Let $z = (z_1, \dots, z_{l+1}) \in \mathbb{C}^{l+1}$.
The difference operator for $BC(l,1)$ introduced in~\cite{Sergeev_2009} 
has the form
\[
D = \sum_{i=1}^l\big(a_{2e_i}(z)(T_{2e_i}-1) + a_{-2e_i}(z)(T_{-2e_i}-1)\big)+a_{2\sqrt{k}e_{l+1}}(z)\big(T_{2\sqrt{k}e_{l+1}}-1\big)+a_{-2\sqrt{k}e_{l+1}}(z)\big(T_{-2\sqrt{k}e_{l+1}}-1\big),
\]
where $T_v$ denotes for any $v \in \mathbb{C}^{l+1}$ the shift operator defined by $T_v f(z) = f(z+v)$ for any function $f(z)$, and
\begin{align}
a_{\pm 2e_i}(z) &= \left(1 \mp \frac{m+2n}{z_i}\right)\left(1 \mp \frac{m}{z_i \pm 1}\right)\prod_{\substack{j=1 \\ j \neq i}}^l \left(1 \mp \frac{2k}{z_i-z_j}\right)\left(1 \mp \frac{2k}{z_i+z_j}\right) \nonumber \\
&\quad \times \left(1 \mp \frac{2}{z_i+\sqrt{k}z_{l+1} \pm (1-k)}\right)\left(1 \mp \frac{2}{z_i-\sqrt{k}z_{l+1} \pm (1-k)}\right), \nonumber
\end{align}

\begin{align*}
a_{\pm 2\sqrt{k}e_{l+1}}(z) &= \frac{1}{k}\left(1 \mp \frac{\sqrt{k}(p+2r)}{z_{l+1}}\right)\left(1 \mp \frac{\sqrt{k}p}{z_{l+1} \pm \sqrt{k}}\right) \\
&\quad \times \prod_{i =1}^l \left(1 \mp \frac{2k}{\sqrt{k}z_{l+1}+z_i \pm (k-1)}\right)\left(1 \mp \frac{2k}{\sqrt{k}z_{l+1}-z_i \pm (k-1)}\right).
\end{align*}
In this and the next section, we are assuming that $m,n,p,r \in \mathbb{Z}_{\geq 0}$, and if $l>1$ then also that $k \in \mathbb{Z}_{>0}$.

In order to prove below in Theorem~\ref{thm: preservation of quasi-invariants} that the operator $D$ preserves the ring $\mathcal{R}_{BC(l,1)}$, we first establish the following lemma.

\begin{lem}\label{analytic}
Let $f \in \mathcal{R}_{BC(l,1)}$. Then $Df$ is analytic.
\begin{proof}
Based on the form of the functions $a_{\pm 2e_i}$ and $a_{\pm 2\sqrt{k}e_{l+1}}$, the only possible singularities of $Df$ are potential simple poles at $z_i=0$ for $1 \leq i \leq l+1$, as well as $z_{l+1} = \pm \sqrt{k}$, $z_i=\pm 1$, $z_i = \pm \sqrt{k}z_{l+1}+k-1$, and $z_i = \pm \sqrt{k}z_{l+1}+1-k$ for $1 \leq i \leq l$, and also $z_i = \pm z_j$ for $1 \leq i < j \leq l$. The strategy is to show that at each of the possible poles,~$Df$ has zero residue. It will follow that $Df$ is analytic everywhere. We describe the computation of the residue for most of the cases, the procedure for the remaining ones being analogous.
Let us denote the residue of a function $f$ at a simple pole at $z_i = c \in \mathbb{C}$ by 
\[
\mathrm{Res}(f,z_i=c)=\lim_{z_i \to c}(z_i-c)f(z).
\]

We may assume, for simplicity, that $m \neq 0 \neq p$, as the case $m=p=0$ was covered in~\cite{Feigin_2005}.
Let us first compute the residue of $Df(z)$ at $z_i=0$ for $1 \leq i \leq l$. For $\alpha \in \mathbb{R}^{l+1} \setminus \{0\}$, let $s_{\alpha}$ denote the orthogonal reflection sending $\alpha$ to $-\alpha$ and fixing pointwise the hyperplane orthogonal to $\alpha$. We note that
\begin{equation}\label{D ei symmetry}
    a_{2e_i}(z)=a_{-2e_i}(s_{e_i}(z)),
\end{equation} which implies
\[
\Res(a_{2e_i}, z_i=0)=-\Res(a_{-2e_i}, z_i=0).
\]
Also, if $m \neq 1$ then
\[
T_{2e_i}f(z)|_{z_i=0}=T_{-2e_i}f(z)|_{z_i=0},
\]
where we use that $f \in \mathcal{R}_{BC(l,1)}$.
Note that if $m=1$ then $\Res(a_{\pm 2e_i}, z_i=0) = 0$. It follows that for any~$m \in \mathbb{Z}_{> 0}$ the residue of $Df(z)$ at $z_i=0$ is zero. Its residue at $z_{l+1}=0$ can be shown to vanish in an analogous way.

Next, we consider $z_i=-1$. The only coefficient function in $D$ that has a pole there is $a_{2e_i}$, and so $\Res(Df, z_i = -1) = 0$ follows from the fact that
\[
(T_{2e_i}f(z)-f(z))|_{z_i=-1}=0,
\]
where we again used that $f \in \mathcal{R}_{BC(l,1)}$. Similarly, the residues of $Df(z)$ at $z_i=1$ and $z_{l+1}= \pm \sqrt{k}$ are zero.

Let now $1 \leq i < j \leq l$, and
let us compute the residue of $Df(z)$ at $z_i=z_j$. We note that 
\begin{equation}\label{D ei-ej symmetry}
    a_{\pm 2e_i}(z)=a_{\pm 2e_j}(s_{e_i-e_j}(z)),
\end{equation} 
which implies 
\[
\Res(a_{\pm 2e_i}, z_i=z_j)=-\Res(a_{ \pm 2e_j}, z_i=z_j).
\]
Also, for any $u \in \mathbb{C}$ and $q = u+1$, we have 
\begin{align*}
T_{2e_i}f(z)|_{z_i=z_j=u}&= 
f(z+e_i-e_j)|_{z_i=z_j=q} =f(z-e_i+e_j)|_{z_i=z_j=q} =T_{2e_j}f(z)|_{z_i=z_j=u},
\end{align*}
where in the penultimate equality we used that $f \in \mathcal{R}_{BC(l,1)}$. Similarly,
$
T_{-2e_i}f(z)=T_{-2e_j}f(z)
$ at $z_i=z_j$.
It follows that the residue of $Df(z)$ at $z_i=z_j$ is zero. The case of $z_i=-z_j$ is similar.

Finally, we consider $z_i=-\sqrt{k}z_{l+1}+k-1$ for $1 \leq i \leq l$. We calculate the following residues: 

\begin{align}
&\Res(a_{2e_i},\, z_i=-\sqrt{k}z_{l+1}+k-1)= 
-2\left(1+\frac{m+2n}{\sqrt{k}z_{l+1}+1-k}\right)\left(1+\frac{m}{\sqrt{k}z_{l+1}-k}\right) \nonumber \\ & \times \left(1+\frac{1}{\sqrt{k}z_{l+1}}\right) \prod_{ \substack{j=1 \\j\neq i}}^l \left(1+\frac{2k}{\sqrt{k}z_{l+1}+z_j+1-k}\right)\left(1+\frac{2k}{\sqrt{k}z_{l+1}-z_j+1-k}\right), \nonumber
\end{align}

\begin{align}
&\Res(a_{-2\sqrt{k}e_{l+1}}, \, z_i=-\sqrt{k}z_{l+1}+k-1)= 2\left(1+\frac{\sqrt{k}(p+2r)}{z_{l+1}}\right)\left(1+\frac{\sqrt{k}p}{z_{l+1}-\sqrt{k}}\right) \nonumber \\ & \times \left(1+\frac{k}{\sqrt{k}z_{l+1}+1-k}\right) \prod_{\substack{j=1 \\ j\neq i}}^l \left(1+\frac{2k}{\sqrt{k}z_{l+1}+z_j+1-k}\right)\left(1+\frac{2k}{\sqrt{k}z_{l+1}-z_j+1-k}\right). \nonumber
\end{align}
In order to compare them, we observe that 
\begin{equation*}
   \frac{\sqrt{k}p}{z_{l+1}-\sqrt{k}} = \frac{m}{\sqrt{k}z_{l+1}-k}, 
\end{equation*}
since $m=kp$; and that moreover, since $2n+1 = k(2r+1)$, we also have
\begin{align*}
\left(1+\frac{\sqrt{k}(p+2r)}{z_{l+1}}\right)\left(1+\frac{k}{\sqrt{k}z_{l+1}+1-k}\right) &= \left(\frac{\sqrt{k}z_{l+1}+m+2n+1-k}{\sqrt{k}z_{l+1}}\right)\left(\frac{\sqrt{k}z_{l+1}+1}{\sqrt{k}z_{l+1}+1-k}\right)  \\
&=\left(1+\frac{m+2n}{\sqrt{k}z_{l+1}+1-k}\right)\left(1+\frac{1}{\sqrt{k}z_{l+1}}\right). 
\end{align*}
Therefore, 
\[
\Res(a_{2e_i}, \,z_i=-\sqrt{k}z_{l+1}+k-1)=-\Res(a_{-2\sqrt{k}e_{l+1}},\,z_i=-\sqrt{k}z_{l+1}+k-1).
\]
Also, for any $u \in \mathbb{C}$ and $q = -u+\sqrt{k}$, we have 
\begin{align*}
T_{2e_i}&f(z)|_{z_i=-\sqrt{k}u+k-1, \, z_{l+1} = u} =
f(z+e_i+\sqrt{k}e_{l+1})|_{z_i=\sqrt{k}q, \, z_{l+1}=-q}  \\
&=f(z-e_i-\sqrt{k}e_{l+1})|_{z_i=\sqrt{k}q, \, z_{l+1}=-q}  =T_{-2\sqrt{k}e_{l+1}}f(z)|_{z_i=-\sqrt{k}u+k-1, \, z_{l+1}=u}, 
\end{align*}
where in the penultimate equality we used that $f \in \mathcal{R}_{BC(l,1)}$.
It follows that the residue of $Df(z)$ at $z_i=-\sqrt{k}z_{l+1}+k-1$ is zero. The cases $z_i=-\sqrt{k}z_{l+1}+1-k$ and $z_i=\sqrt{k}z_{l+1} \pm (1-k)$ are similar. This completes the proof.
    \end{proof}
\end{lem}

\begin{thm}\label{thm: preservation of quasi-invariants}
Let $f \in \mathcal{R}_{BC(l,1)}$. Then $Df \in \mathcal{R}_{BC(l,1)}$. 
\begin{proof}
By Lemma \ref{analytic}, $Df$ is analytic, so it only remains to show that $Df$ satisfies the functional identities~\eqref{quasi}.
Let $\alpha \in BC(l,1)_+^r$ and $s \in A_\alpha$.

Suppose $\alpha = e_i$ for $1 \leq i \leq l$. Then at $\langle \alpha, z \rangle = 0$, we have $a_{2e_i}(z+s\alpha)=a_{-2e_i}(z-s\alpha)$  by equality~\eqref{D ei symmetry}. And for all $j \neq i$, it is straightforward to see that $a_{\pm 2e_j}(z+s\alpha)=a_{\pm 2e_j}(z-s\alpha)$ and $a_{\pm 2\sqrt{k}e_{l+1}}(z+s\alpha)=a_{\pm 2\sqrt{k}e_{l+1}}(z-s\alpha)$ at $\langle \alpha, z \rangle = 0$. For $s \neq 1$, we have $a_{-2e_i}(z+s\alpha) = a_{2e_i}(z-s\alpha)$ at $\langle \alpha, z \rangle = 0$ by equality~\eqref{D ei symmetry} (for $s=1$, the functions $a_{-2e_i}(z+s\alpha)$ and $a_{2e_i}(z-s\alpha)$ are singular at~$\langle \alpha, z \rangle = 0$ and we will deal with this case separately).

Now observe that for $s \neq m-1, \, m+2n,$ we have $s+2 \in A_{e_i}$ and thus since $f \in \mathcal{R}_{BC(l,1)}$ that $T_{2e_i}f(z+se_i)|_{z_i=0}=T_{-2e_i}f(z-se_i)|_{z_i=0}$. On the other hand, if $s \in \{m-1, m+2n\}$ then $a_{2e_i}(z+se_i)|_{z_i=0}=0$ $=a_{-2e_i}(z-se_i)|_{z_i=0}$. 
Also, since $\langle z \pm 2e_j, e_i \rangle = \langle z \pm 2\sqrt{k}e_{l+1}, e_i \rangle = z_i$, we have for any $s \in A_{e_i}$ that
 $T_{\pm 2e_j}f(z+se_i)|_{z_i=0}=T_{\pm 2e_j}f(z-se_i)|_{z_i=0}$ and $T_{\pm 2\sqrt{k}e_{l+1}}f(z+se_i)|_{z_i=0}=T_{\pm 2\sqrt{k}e_{l+1}}f(z-se_i)|_{z_i=0}$. And for $s \neq 1$, we have $s-2 \in A_{e_i} \cup \{0\}$ and hence
 $T_{-2e_i}f(z+se_i)|_{z_i=0}=T_{2e_i}f(z-se_i)|_{z_i=0}$. It follows from this and the previous paragraph that the identities~\eqref{quasi} are satisfied for $Df$ for $\alpha = e_i$ for any $s \in A_{e_i} \setminus \{1\}$.

 Let us now deal with the case $s=1$. The property that
 $$\left(f(z-e_i) - f(z+e_i)\right)|_{z_i=0} = 0$$
 can be restated as $f(z-e_i) - f(z+e_i) = z_i g(z)$ for some analytic function $g(z)$. Thus, we have
 \begin{align*}
     \lim_{z_i \to 0} a_{-2e_i}(z&+e_i)(T_{-2e_i}-1)f(z+e_i) = 
     g(z)|_{z_i=0}\lim_{z_i \to 0} z_ia_{-2e_i}(z+e_i) \\
     &\stackrel{\eqref{D ei symmetry}}{=} 
     -  g(z)|_{z_i=0}\lim_{z_i \to 0} z_i a_{2e_i}(z-e_i) 
     =  \lim_{z_i \to 0} a_{2e_i}(z-e_i)(T_{2e_i} - 1)f(z-e_i).
 \end{align*}
It follows from this and the previous two paragraphs that the identity~\eqref{quasi} for $Df$ holds for $\alpha = e_i$ also when~$s=1$. 
The case of $\alpha = \sqrt{k}e_{l+1}$ can be dealt with in an analogous way.

Next, suppose $\alpha = e_i-e_j$ for $1 \leq i < j \leq l$. Then
$
a_{\pm 2e_i}(z+ \widetilde{s}\alpha)=a_{\pm 2e_j}(z-\widetilde{s}\alpha)
$ at $\langle \alpha, z \rangle = 0$ for $\widetilde{s} \in \{ \pm s\}$ by equality~\eqref{D ei-ej symmetry}. 
If $s \neq k$ then $s+1 \in A_{e_i - e_j}$ and then for any $u \in \mathbb{C}$ and $q = u+1$, we have
\begin{align*}
T_{2e_i}f(z +s(e_i-e_j))|_{z_i= z_j=u} &= f(z + (s+1)(e_i-e_j))|_{z_i = z_j=q} \\
&= f(z - (s+1)(e_i-e_j))|_{z_i = z_j=q} = T_{2e_j}f(z -s(e_i-e_j))|_{z_i= z_j=u}, 
\end{align*}
where we used that $f \in \mathcal{R}_{BC(l,1)}$;
and similarly 
$T_{-2e_j} f(z+s\alpha) = T_{-2e_i}f(z-s\alpha)$ at $\langle \alpha, z \rangle = 0$.
On the other hand, if $s=k$ then $a_{2e_i}(z+s\alpha) = 0 =a_{-2e_j}(z+s\alpha)$ at $\langle \alpha, z \rangle = 0$.
Moreover, for any $s \in A_{e_i -e_j}$, we have $s-1 \in A_{e_i -e_j} \cup \{0\}$, which can be used to see that 
$T_{-2e_i}f(z+s\alpha) = T_{-2e_j}f(z-s\alpha)$ and $T_{2e_j}f(z+s\alpha) = T_{2e_i}f(z-s\alpha)$ at $\langle \alpha, z \rangle = 0$.

For all $t \neq i, j$ and $\langle \alpha, z \rangle = 0$, 
it is straightforward to see that $a_{\pm 2e_t}(z+s\alpha)=a_{\pm 2e_t}(z-s\alpha)$ and $a_{\pm 2 \sqrt{k}e_{l+1}}(z+s\alpha) = a_{\pm 2 \sqrt{k}e_{l+1}}(z-s\alpha)$; and that 
$T_{\pm 2e_t}f(z+s\alpha) = T_{\pm 2 e_t}f(z-s\alpha)$ and $T_{\pm 2 \sqrt{k}e_{l+1}}f(z+s\alpha) = T_{\pm 2 \sqrt{k}e_{l+1}}f(z-s\alpha)$ since $f \in \mathcal{R}_{BC(l,1)}$ and $\langle z \pm 2e_t, \alpha \rangle = \langle z \pm 2\sqrt{k}e_{l+1}, \alpha \rangle = 0$.

It follows from the above two paragraphs that the identities~\eqref{quasi} are satisfied for $Df$ for $\alpha = e_i-e_j$. The case of $\alpha = e_i + e_j$ is analogous.

Finally, suppose $\alpha = e_i + \sqrt{k}e_{l+1}$ for $1 \leq i \leq l$. Note that $A_{e_i + \sqrt{k}e_{l+1}} = \{1\}$, so $s=1$. For $\varepsilon \in \{\pm 1\}$, 
\[
a_{2\varepsilon e_i}(z + \varepsilon (e_i + \sqrt{k}e_{l+1}))|_{z_i=-\sqrt{k}z_{l+1}} = 0 = a_{2\varepsilon\sqrt{k}e_{l+1}}(z + \varepsilon (e_i + \sqrt{k}e_{l+1}))|_{z_i=-\sqrt{k}z_{l+1}}.
\]
We also have
\[
a_{- 2\varepsilon e_i}(z + \varepsilon (e_i + \sqrt{k}e_{l+1}))|_{z_i=-\sqrt{k}z_{l+1}}=a_{2\varepsilon \sqrt{k}e_{l+1}}(z - \varepsilon (e_i + \sqrt{k}e_{l+1}))|_{z_i=-\sqrt{k}z_{l+1}}.
\]
The latter can be seen upon rewriting $m$ and $n$ in terms of $p$ and $r$, and using, in particular, that
\begin{equation*}
    \left(1 - \frac{\varepsilon(m+2n)}{\sqrt{k}z_{l+1} -\varepsilon}\right) 
    \left(1 - \frac{\varepsilon}{\sqrt{k}z_{l+1}}\right) = \frac{\sqrt{k}z_{l+1} -\varepsilon k -\varepsilon k\left(p+2r\right)}{\sqrt{k}z_{l+1}} 
    = \left(1 - \frac{\varepsilon\sqrt{k}\left(p+2r\right)}{z_{l+1} - \varepsilon\sqrt{k}}\right)
    \left(1 - \frac{\varepsilon\sqrt{k}}{z_{l+1}}\right).
\end{equation*}
Moreover, we have
\[
\begin{aligned}
(T_{-2\varepsilon e_i}-1)f(z + \varepsilon (e_i + \sqrt{k}e_{l+1}))|_{z_i=-\sqrt{k}z_{l+1}} &= \left( f(z - \varepsilon (e_i - \sqrt{k}e_{l+1})) -  f(z + \varepsilon (e_i + \sqrt{k}e_{l+1}))\right)|_{z_i=-\sqrt{k}z_{l+1}} \\
&=(T_{2\varepsilon\sqrt{k}e_{l+1}}-1)f(z - \varepsilon (e_i + \sqrt{k}e_{l+1}))|_{z_i=-\sqrt{k}z_{l+1}},
\end{aligned}
\]
where we used that $f \in \mathcal{R}_{BC(l,1)}$.

For $j \neq i$ and $z_i=-\sqrt{k}z_{l+1}$, we have 
$a_{2e_j}(z + e_i + \sqrt{k}e_{l+1})=a_{2e_j}(z - e_i - \sqrt{k}e_{l+1})$,
which can be seen by using that
\begin{align*}
    &\left(1 - \frac{2k}{z_j + \sqrt{k}z_{l+1} - 1} \right)
    \left(1 - \frac{2k}{z_j - \sqrt{k}z_{l+1} + 1} \right)
    \left(1 - \frac{2}{z_j + \sqrt{k}z_{l+1} + 1} \right)
    \left(1 - \frac{2}{z_j - \sqrt{k}z_{l+1} + 1 - 2k} \right) \\
    &= \frac{(z_j + \sqrt{k}z_{l+1} - 1-2k)(z_j - \sqrt{k}z_{l+1} - 1 - 2k)}{(z_j - \sqrt{k}z_{l+1} + 1)(z_j + \sqrt{k}z_{l+1} + 1)}\\
    &= 
    \left( 1-\frac{2k}{z_j + \sqrt{k}z_{l+1} + 1} \right)
    \left( 1-\frac{2k}{z_j - \sqrt{k}z_{l+1} - 1} \right)
    \left( 1-\frac{2}{z_j + \sqrt{k}z_{l+1} + 1 - 2k}\right)
    \left( 1-\frac{2}{z_j - \sqrt{k}z_{l+1} + 1}\right);
\end{align*}
and it is a similar calculation to show that
$a_{-2e_j}(z + e_i + \sqrt{k}e_{l+1})=a_{-2e_j}(z - e_i - \sqrt{k}e_{l+1})$. Moreover, 
$T_{2\varepsilon e_j}f(z + e_i + \sqrt{k}e_{l+1}) = T_{2\varepsilon e_j}f(z - e_i - \sqrt{k}e_{l+1})$ at 
$z_i=-\sqrt{k}z_{l+1}$
since $f \in \mathcal{R}_{BC(l,1)}$.

It follows that the identities~\eqref{quasi} are satisfied for $Df$ for $\alpha = e_i +\sqrt{k}e_{l+1}$. The case of $\alpha = e_i -\sqrt{k}e_{l+1}$ is analogous.
This completes the proof.
    \end{proof}
\end{thm}

\section{Construction of the Baker--Akhiezer function}\label{sec: BA function}
Even though the configuration $BC(l,1)_+$ is non-reduced, all of its subsets of collinear vectors are of the form~$\{\alpha$, $2 \alpha\}$. Thus, we can use the axiomatic definition of BA functions proposed in~\cite[Definition 2.1]{Feigin_2022} 
(see~\cite{Chalykh_1990, Chalykh_1993} for the original axiomatic definition of BA functions formulated by Chalykh,  Styrkas, and Veselov for arbitrary finite collections of \emph{non-collinear} vectors with
integer multiplicities, and also~\cite{Feigin_2005} for a weaker version of that axiomatics).

\begin{defn}\label{def:BA}
    A function $\psi\colon \mathbb{C}^{l+1} \times \mathbb{C}^{l+1} \to \mathbb{C}$ is a BA function for the configuration~$BC(l,1)$ with non-negative integer multiplicities if it satisfies the following conditions:
    \begin{enumerate}
    \begin{item}
$\psi(z,x) = P(z,x)e^{\langle z,x \rangle}$ for a polynomial $P$ in $z$ with highest-order term $\prod_{\alpha \in BC(l,1)_+}\langle \alpha, z \rangle ^{m_{\alpha}}$,
    \end{item}
    \begin{item}
$\psi(z+s\alpha,x) = \psi(z-s\alpha,x)$ at $\langle z,\alpha \rangle = 0 $ for all $\alpha \in BC(l,1)^r_+$ and $s  \in A_\alpha$.
    \end{item}
    \end{enumerate}
\end{defn}
Note that condition 2 in Definition~\ref{def:BA} is similar to the functional identities~\eqref{quasi} satisfied by elements of the ring~$\mathcal{R}_{BC(l,1)}$; and indeed for any $x \in \mathbb{C}^{l+1}$ for which $P(z,x)$ is non-singular, the function $\psi_x\colon z \mapsto \psi (z,x)$ belongs to~$\mathcal{R}_{BC(l,1)}.$

It follows from the general results proved in~\cite[Proposition 2.7, Theorem 2.8]{Feigin_2022} that if a function $\psi$ satisfying Definition~\ref{def:BA} exists then it is unique, and it is a joint eigenfunction for a large commutative ring of differential operators in~$x$. By~\cite[Proposition 2.10]{Feigin_2022}, this ring contains the Hamiltonian~\eqref{eqn: BC(l,1) CMS Hamiltonian} for $\mathcal{A} = BC(l,1)_+$ with the non-negative integer values of the multiplicity parameters. Moreover, this ring also contains a complete set of (quantum) integrals for this Hamiltonian, as well as extra integrals that correspond to the algebraic integrability of this system (see~\cite{Chalykh_1990, Chalykh_1993} for the definition of algebraic integrability). Namely, for every polynomial $p(z) \in \mathcal{R}_{BC(l,1)}$, there is a differential operator in~$x$ that commutes with the Hamiltonian and whose highest symbol is~$p_0(\partial_x)$, where $p_0$ is the highest homogeneous term of $p$. 

The following theorem gives an explicit construction of the BA function for~$BC(l,1)$ by using the difference operator $D$ from the previous section. The BA function will be an eigenfunction for
the operator $D$, which shows bispectrality of the generalised CMS 
Hamiltonian for $BC(l,1)$ for non-negative integer values of the multiplicity parameters (the case of non-integer multiplicities will be considered in the next section). See~\cite{Chalykh_2000, Feigin_2005} and~\cite[Theorems 5.3 and 6.3]{Feigin_2022} for earlier examples of analogous results obtained by similar methods for other configurations. 

\begin{thm}\label{thm: bispectrality int case}
    Let $M = \sum_{\alpha \in BC(l,1)_ +}m_\alpha=l(m+n+(l-1)k+2)+p+r$, and let $S=\{\pm 2e_i, \pm 2\sqrt{k}e_{l+1} \,|$ $1 \leq i \leq l\}.$ For $x \in \mathbb{C}^{l+1}$, let
    \[
    \mu (x) = \sum_{\tau \in S}\kappa_{\tau}(e^{\langle \tau, x\rangle}-1),
    \]
    and
    \begin{equation}\label{eqn: c(x)}
        c(x)=\frac{M!}{2^{ln+r}}\prod_{\alpha \in BC(l,1)_{+}^r}\left(\sum_{\tau \in S}\kappa_{\tau} \langle \tau, \alpha \rangle e^{\langle \tau, x \rangle}\right)^{m_{\alpha}+m_{2\alpha}},
    \end{equation}
    where $\kappa_{\tau} = k^{-1}$ if $\tau = \pm 2\sqrt{k}e_{l+1}$, and $\kappa_{\tau} =1$ otherwise. For $z \in \mathbb{C}^{l+1}$, let $Q(z)$ be the polynomial in~$\mathcal{R}_{BC(l,1)}$ given by 
    \[
    Q(z) = \prod_{\substack{\alpha \in BC(l,1)_+^r \\ s \in A_\alpha}}\left(\langle \alpha,z\rangle^2-s^2\langle\alpha,\alpha \rangle^2\right).
    \]
    Then the function
    \begin{equation}\label{eqn: BA function}
        \psi(z,x)=c(x)^{-1}\left(D-\mu(x)\right)^M[Q(z)e^{\langle z,x \rangle}]
    \end{equation}
    is the BA function for $BC(l,1)$. Moreover, $\psi$ is an eigenfunction of the operator $D$ with $$D\psi(z,x)=\mu(x)\psi(z,x).$$
\end{thm}

The proof is analogous to the case of $AG_2$~\cite[Theorem 5.3]{Feigin_2022}, it just uses the following two lemmas in place of \cite[Lemma 5.1]{Feigin_2022} and \cite[Lemma 5.2]{Feigin_2022}, respectively. 

\begin{lem}\label{lem: analogue of 5.1}
    For all $\tau \in S = \{\pm 2e_i, \pm 2\sqrt{k}e_{l+1} \,|\, 1 \leq i \leq l\}$, the coefficient function $a_\tau$ in the operator $D$ can be expanded as 
    \begin{equation*}
        a_\tau(z) = \kappa_\tau - \kappa_\tau \sum_{\alpha \in BC(l,1)^r_+} \frac{\langle \tau, \alpha \rangle (m_\alpha + m_{2\alpha})}{\langle \alpha, z \rangle} + R_\tau(z),
    \end{equation*}
    where $\kappa_\tau$ are as in Theorem~\ref{thm: bispectrality int case}, and $R_\tau$ is a rational function with $\deg R_\tau \leq -2$.
\end{lem}
The above way of expanding the coefficients of the operator $D$ can be used to prove the following lemma (in a completely analogous way to how \cite[Lemma 5.2]{Feigin_2022} is proved by using \cite[Lemma 5.1]{Feigin_2022}).

\begin{lem}\label{lem: analogue of 5.2}
    For $\alpha \in BC(l,1)^r_+$, let $n_\alpha \in \mathbb{Z}_{\geq 0}$ be arbitrary. Let $N = \sum_{\alpha \in BC(l,1)^r_+} n_\alpha$. Let $\mu(x)$ be as in Theorem~\ref{thm: bispectrality int case}, and let $A(z) = \prod_{\alpha \in BC(l,1)^r_+} \langle \alpha, z \rangle^{n_\alpha}$. Then we have $(D-\mu)[A(z) e^{\langle z,x \rangle}] = R(z,x)e^{\langle z, x \rangle}$ for a rational function~$R(z,x)$ in $z$ of the form
    \begin{equation*}
        R(z,x) = \sum_{\alpha \in BC(l,1)^r_+}(n_\alpha - m_\alpha - m_{2\alpha})\left(\sum_{\tau \in S}\kappa_\tau \langle \tau, \alpha \rangle e^{\langle \tau, x \rangle} \right)A(z) \langle \alpha, z \rangle^{-1} + S(z,x),
    \end{equation*}
   where $S(z,x)$ is a rational function in $z$ of degree less than or equal to $N-2$.

   In particular, for \emph{any} polynomial $B(z,x)$ in $z$, we have that $(D-\mu)[B(z,x) e^{\langle z,x \rangle}] = U(z,x)e^{\langle z, x \rangle}$ for a rational function~$U(z,x)$ in $z$ with $\deg U(z,x) \leq \deg B(z,x) - 1$.
\end{lem}

\begin{proof}[Proof of Theorem~\ref{thm: bispectrality int case}]
The idea of the proof is as follows. We have $Q(z)e^{\langle z,x \rangle} \in \mathcal{R}_{BC(l,1)}$ because $Q(z+s\alpha)=Q(z-s\alpha)=0$ at $\langle z, \alpha \rangle = 0$ for all $\alpha \in BC(l,1)^r_+$ and $s \in A_\alpha$. The property that $(D-\mu)^M[Q(z)e^{\langle z,x \rangle}]$ satisfies condition~2 in Definition~\ref{def:BA} thus follows from Theorem~\ref{thm: preservation of quasi-invariants}. Moreover, each repeated application of $D-\mu$ on $Q(z)e^{\langle z,x \rangle}$ gives a function of the form $R(z,x)e^{\langle z,x \rangle}$ with $R(z,x)$ a polynomial in $z$, which follows from the form of the operator~$D$ and Theorem~\ref{thm: preservation of quasi-invariants}. More specifically, for all $b \in \mathbb{Z}_{>0}$, we have that $(D-\mu)^b[Q(z)e^{\langle z,x \rangle}]=R_b(z,x)e^{\langle z,x \rangle}$ where $R_b(z,x)$ is a polynomial in $z$ of degree at most $2M-b$ whose highest-order homogeneous component can be kept track of by using Lemma~\ref{lem: analogue of 5.2} (similarly to how this is done in the case of~$AG_2$ in the proof of \cite[Theorem 5.3]{Feigin_2022} by using \cite[Lemma 5.2]{Feigin_2022}). This makes it possible to see that $(D-\mu)^M[Q(z)e^{\langle z,x \rangle}]$ essentially satisfies also condition~1 in Definition~\ref{def:BA}, except that the highest-order term of its polynomial part has an extra factor of~$c(x)$ 
given by formula~\eqref{eqn: c(x)}. It follows that $\psi$ defined by the expression~\eqref{eqn: BA function} is the BA function.
At the next application of $D-\mu$, we get $(D-\mu)^{M+1}[Q(z)e^{\langle z,x \rangle}]=0$ as a consequence of~\cite[Lemma 2.6]{Feigin_2022}, which implies that $D\psi = \mu \psi$.
\end{proof}

In the next section, we will need some further properties of the BA function~\eqref{eqn: BA function}, which we record in 
Propositions~\ref{prop: analogue of 6.5} and~\ref{prop: analogue of 6.6} below (cf.~\cite[Propositions 6.5 and 6.6]{Chalykh_2000}, respectively). The proof of the former one uses the following two lemmas. 
Let $\propto$ denote proportionality by a constant factor. 
\begin{lem}\label{lem: c}
    The function $c(x)$ defined by formula~\eqref{eqn: c(x)} satisfies
    \begin{align*}
        c(x) \propto \prod_{\alpha \in BC(l,1)^r_+} \sinh^{m_\alpha+m_{2\alpha}} \langle 2\alpha, x \rangle.
    \end{align*}
\end{lem}
\begin{proof}
    Using formula~\eqref{eqn: c(x)}, we compute
    \begin{align*}
        c(x)&(e^{2\sqrt{k}x_{l+1}}-e^{-2\sqrt{k}x_{l+1}})^{-p-r} 
        \prod_{i=1}^l(e^{2x_i}-e^{-2x_i})^{-m-n}  \propto  \\
        &\prod_{\substack{1 \leq i \leq l \\ \varepsilon \in \{\pm 1\}}} (e^{2x_i}-e^{-2x_i} + e^{2\varepsilon\sqrt{k}x_{l+1}}-e^{-2\varepsilon\sqrt{k}x_{l+1}}) 
        \prod_{\substack{1 \leq i<j \leq l \\ \varepsilon \in \{\pm 1\}}}
        (e^{2x_i}-e^{-2x_i}+e^{2\varepsilon x_j}-e^{-2\varepsilon x_j})^k \\
        &= 
        \prod_{\substack{1 \leq i \leq l \\ \varepsilon \in \{\pm 1\}}} (e^{x_i+ \varepsilon\sqrt{k}x_{l+1}}+e^{-x_i-\varepsilon\sqrt{k}x_{l+1}})
        (e^{x_i+\varepsilon\sqrt{k}x_{l+1}}-e^{-x_i-\varepsilon\sqrt{k}x_{l+1}}) \\
        &\qquad \times
        \prod_{\substack{1 \leq i < j \leq l \\ \varepsilon \in \{\pm 1\}}} (e^{x_i+\varepsilon x_j}+e^{-x_i- \varepsilon x_j})^k 
        (e^{x_i+\varepsilon x_j}-e^{-x_i-\varepsilon x_j})^k, 
    \end{align*}
    where we used the identity
    $A^2 - A^{-2} + B^2 - B^{-2} = (A B^{-1} + A^{-1}B)(AB - A^{-1}B^{-1})$.
    Then by using the difference of two squares formula, we get 
    \begin{equation*}
        c(x) \propto 
        \prod_{\alpha \in BC(l,1)^r_+}(e^{\langle 2\alpha, x \rangle}-e^{-\langle 2\alpha, x \rangle})^{m_\alpha + m_{2\alpha}} \propto \prod_{\alpha \in BC(l,1)^r_+} \sinh^{m_\alpha+m_{2\alpha}} \langle 2\alpha, x \rangle,
    \end{equation*}
    as required.
\end{proof}

Let us consider the following function
\begin{equation}\label{eqn: delta(x)}
    \delta(x) = \prod_{\alpha \in BC(l,1)_+} (2 \sinh\langle \alpha, x \rangle)^{m_\alpha}. 
\end{equation}
The next lemma relates $\delta(x)$ to the function $c(x)$ from above. 

\begin{lem}\label{lem: c vs delta}
    We have
    \begin{align*}
        &c(x) \propto \delta(x)\prod_{\alpha \in BC(l,1)^r_+} \cosh^{m_\alpha} \langle \alpha, x \rangle. 
    \end{align*}
\end{lem}
\begin{proof}
    Note that
    \begin{align}
        \delta(x) 
        \propto \prod_{\alpha \in BC(l,1)^r_+} \sinh^{m_\alpha + m_{2\alpha}} \langle \alpha, x \rangle \, \cosh^{m_{2\alpha}} \langle \alpha, x \rangle. \label{eqn: delta(x) expanded}
    \end{align}
    The proof is thus completed by making use of Lemma~\ref{lem: c}.
\end{proof}

The proof of the next proposition is based on the ideas of the proof of~\cite[Proposition 6.5]{Chalykh_2000}. Let us define the following lattice 
    \begin{equation}\label{eqn: lattice}
        \mathcal{L} = \mathcal{L}(k) = 2\mathbb{Z}e_1 \oplus \cdots \oplus 2\mathbb{Z}e_l \oplus 2\sqrt{k}\mathbb{Z}e_{l+1}.
    \end{equation}
Let $\mathcal{L}_+$ be the semigroup of $\nu = (\nu_1, \dots, \nu_l, \sqrt{k} \nu_{l+1}) \in \mathcal{L}$ with non-negative partial sums of $\nu_i$, that is
$$
    \mathcal{L}_+ = \{\nu = (\nu_1, \dots, \nu_l, \sqrt{k} \nu_{l+1}) \in \mathcal{L} \, | \, \textstyle \sum_{i=1}^r \nu_i \geq 0  \text{ for } r=1, \dots, l+1 \} = \bigoplus_{i=1}^{l+1} \mathbb{Z}_{\geq 0} \alpha_i,
$$
where $\alpha_1 = 2(e_1-e_2)$, \textellipsis, $\alpha_{l-1} = 2(e_{l-1}-e_l)$, $\alpha_l = 2(e_l - \sqrt{k}e_{l+1})$, and $\alpha_{l+1} = 2\sqrt{k}e_{l+1}$.
We note that $2\alpha \in \mathcal{L}_+$ for all $\alpha \in BC(l,1)_+$.

\begin{prop}\label{prop: analogue of 6.5}
    The function $\psi$ defined by formula~\eqref{eqn: BA function} can be expanded in the form
    \begin{equation}\label{eqn: psi expanded}
        \psi = \delta(x)^{-1} e^{\IP{z-\rho}{x}} \sum_{\nu \in \mathcal{L}_+} c_\nu(z) e^{\IP{\nu}{x}} 
    \end{equation}
    for some polynomials $c_\nu(z)$, 
    where 
    \begin{equation}\label{eqn: rho}
        \rho = \sum_{\alpha \in BC(l,1)_+} m_\alpha \alpha
    \end{equation}
    and $\delta(x)$ is defined by formula~\eqref{eqn: delta(x)}.
\end{prop}
\begin{proof}
    By \cite[Theorem 2.8, Proposition 2.10]{Feigin_2022}, the BA function $\psi$ satisfies the eigenfunction equation $H\psi = -z^2 \psi$ for the operator $H$ given by formula~\eqref{eqn: BC(l,1) CMS Hamiltonian} with~$\mathcal{A} = BC(l,1)_+$. The operator can be rearranged as follows
    \begin{align*}
        H &= -\Delta + \sum_{\alpha \in BC(l,1)^r_+} \frac{(m_\alpha + m_{2\alpha})(m_\alpha + m_{2\alpha} + 1)\langle \alpha, \alpha \rangle}{\sinh^2 \langle \alpha, x \rangle} 
        - \sum_{\alpha \in BC(l,1)^r_+} \frac{m_{2\alpha}(m_{2\alpha} + 1)\langle \alpha, \alpha \rangle}{\cosh^2 \langle \alpha, x \rangle}.
    \end{align*}
    This form of the operator $H$ makes it possible to see, by Laurent expanding the eigenfunction $\psi$ in $x$ around suitable hyperplanes, that $\psi$ must have either a pole of order $m_\alpha + m_{2\alpha}$ or a zero of order $m_\alpha + m_{2\alpha} + 1$ along each of the hyperplanes $\sinh \IP{\alpha}{x} = 0$ for $\alpha \in BC(l,1)^r_+$; and similarly a pole of order $m_{2\alpha}$ or a zero of order $m_{2\alpha} + 1$ along the hyperplanes $\cosh \IP{\alpha}{x}=0$ for $\alpha \in BC(l,1)^r_+$. 
    The expression for $c(x)$ given in Lemma~\ref{lem: c} suggests that the order of the poles of $\psi$ at $\cosh\IP{\alpha}{x}=0$ might be higher than $m_{2\alpha}$, but the local expansion of~$\psi$ and the eigenvalue equation $H\psi = -z^2 \psi$ imply that this cannot happen. It follows from the form of $c(x)$ that $\psi$ cannot have any additional other singularities either. 
    Hence $\delta(x) \psi$ 
    is analytic (in both $x$ and $z$ variables) due to the property~\eqref{eqn: delta(x) expanded}.
    
    By construction, $\psi(z,x) = c(x)^{-1} \Phi(z,x)$ for the quasi-polynomial in $z$ function
    $$\Phi(z,x) = (D-\mu(x))^M[Q(z) e^{\langle z,x \rangle}].$$
    From the fomulas for $D$ and $\mu$, it is clear that
    $\Phi$ is analytic in $x$ and that it can be expanded as
    \begin{equation}\label{eqn: Phi}
        \Phi = e^{\IP{z}{x}} \sum_{\nu \in \mathcal{L}} b_\nu(z) e^{\IP{\nu}{x}}
    \end{equation}
    for some polynomials $b_\nu(z)$.
    
In view of Lemma~\ref{lem: c vs delta}, the analyticity of $\delta(x) \psi$ implies that the trigonometric polynomial in~$x$ given by~\eqref{eqn: Phi} must be divisible by
    \begin{equation}\label{b(x)}
        \prod_{\alpha \in BC(l,1)^r_+} \cosh^{m_\alpha} \langle \alpha, x \rangle \propto
        e^{\IP{\rho^r}{x}} \prod_{\alpha \in BC(l,1)^r_+}(e^{-2\IP{\alpha}{x}}+1)^{m_\alpha}
        = e^{\IP{\rho}{x}} \sum_{\nu \in \mathcal{L}} d_\nu e^{\IP{\nu}{x}},
    \end{equation}
    where $d_\nu \in \mathbb{R}$ and 
    \begin{equation*}
        \rho^r \coloneqq \sum_{\alpha \in BC(l,1)^r_+} m_\alpha \alpha,
    \end{equation*}
    and we used that $2\alpha \in \mathcal{L}$ for all $\alpha \in BC(l,1)^r_+$.  
    The quotient of the function~\eqref{eqn: Phi} by its divisor~\eqref{b(x)} will still be a trigonometric polynomial in $x$, and we get that
    \begin{equation*}
        \psi = \delta(x)^{-1} e^{\IP{z-\rho}{x}} \sum_{\nu \in \mathcal{L}} c_{\nu}(z) e^{\IP{\nu}{x}} 
    \end{equation*}
    for some polynomials $c_\nu(z)$. Let $\mathcal{P} = \{\nu \in \mathcal{L} \, | \, c_{\nu} \neq 0 \}$. We need to show that $\mathcal{P} \subset \mathcal{L}_+$.
    
    Let $y_i = e^{\IP{\alpha_i}{x}}$ for $i=1, \dots, l+1$. The potential in the Hamiltonian $H$ can be written as
    \begin{equation}\label{eqn: potential}
        \sum_{\alpha \in BC(l,1)_+} \frac{4 m_\alpha(m_\alpha + 2 m_{2\alpha} + 1)\langle \alpha, \alpha \rangle e^{2\IP{\alpha}{x}}}{(e^{2\IP{\alpha}{x}} - 1)^2}. 
    \end{equation}
    By using that $2\alpha \in \mathcal{L}_+ = \oplus_{i=1}^{l+1} \mathbb{Z}_{\geq 0} \alpha_i$ for any $\alpha \in BC(l,1)_+$, we can rewrite the potential~\eqref{eqn: potential} in terms of the variables $y_i$. We can then expand it (for small $y_i$, that is for $x$ in $\{x \in \mathbb{C}^{l+1} \, | \, \re\IP{\alpha_i}{x} < 0, \ i=1, \dots, l+1 \}$) into a Taylor series in $y_i$, which will have no constant term, and thus  obtain 
    \begin{equation*}
        H = -\Delta + \sum_{\mu \in \mathcal{L}_+ \setminus \{0\}} g_\mu e^{\IP{\mu}{x}}
    \end{equation*}
    for some constants $g_\mu$. Similarly, one can expand the function
    \begin{equation*}
        \delta(x)^{-1} = e^{\IP{\rho}{x}}\prod_{\alpha \in BC(l,1)_+} (e^{\IP{2\alpha}{x}} - 1)^{-m_\alpha} 
        = e^{\IP{\rho}{x}} \left( (-1)^t + \sum_{\eta \in \mathcal{L}_+ \setminus \{0\}} h_\eta e^{\IP{\eta}{x}} \right),
    \end{equation*}
    where $h_\eta$ are constants and $t = l(m+n) + p + r$. Thus, the eigenfunction equation $(H+z^2)\psi  = 0$ gives that
    \begin{equation}\label{eqn: SE}
        \begin{aligned}
            \sum_{\nu \in \mathcal{P}} &c_{\nu}(z) e^{\IP{\nu}{x}}
            \bigg(\left(z^2-(z+\nu)^2\right)(-1)^t  + \sum_{\eta \in \mathcal{L}_+ \setminus \{0\}} h_\eta \left(z^2 - (z + \nu +\eta)^2 \right) e^{\IP{\eta}{x}} 
           \\
            & + \sum_{\mu \in \mathcal{L}_+ \setminus \{0\}} (-1)^t g_\mu e^{\IP{\mu}{x}}  + \sum_{\mu, \eta \in \mathcal{L}_+ \setminus \{0\}} g_\mu h_\eta e^{\IP{\mu + \eta}{x}}\bigg) 
            = 0.
        \end{aligned}
    \end{equation}
    Since the set $\mathcal{P}$ is finite, it contains (one or several) minimal elements $\nu_{\min}$ with respect to the partial order on $\mathcal{L}$ defined by $\alpha > \beta$ if and only if $\alpha - \beta \in \mathcal{L}_+ \setminus \{0\}$ for $\alpha, \beta \in \mathcal{L}$.
    The term $e^{\IP{\nu_{\min}}{x}}$ appears only once in the left-hand side of equality~\eqref{eqn: SE}, and its coefficient must hence vanish. We get that $(z+\nu_{\min})^2 = z^2$ for generic $z$, and it follows that $\nu_{\min} = 0$.
    Suppose that there is some~$\nu \in \mathcal{P} \setminus \mathcal{L}_+ $. Then $\nu \neq 0$, so it cannot be minimal in $\mathcal{P}$, hence there must be some $\nu_2 \in \mathcal{P} \setminus \mathcal{L}_+$ with $\nu > \nu_2$.
    By iterating this argument, we get an infinite chain $\nu > \nu_2 > \nu_3 > \dots$ of elements in $\mathcal{P}$, contradicting the finiteness of this set. It follows that $\mathcal{P} \subset \mathcal{L}_+$, as required.
\end{proof}

\begin{rem}
    By adapting the argument in the last paragraph of the proof of~\cite[Proposition 6.5]{Chalykh_2000}, one could show that the set $\mathcal{P}$ from the proof of Proposition~\ref{prop: analogue of 6.5} is contained in the subset of $\mathcal{L}_+$ given by 
    $\{2\textstyle\sum_{\alpha \in BC(l,1)_+} t_\alpha \alpha \, | \, t_\alpha \in \mathbb{Z}, \, 0 \leq t_\alpha \leq m_\alpha \}$.
\end{rem}

The proof of the following proposition is based on the ideas of the proof of~\cite[Proposition 6.6 (2)]{Chalykh_2000}.
\begin{prop}\label{prop: analogue of 6.6}
    The polynomial $c_0$ in the expansion~\eqref{eqn: psi expanded} is given by
    \begin{equation*}
        c_0(z) = (-1)^{l(m+n) + p + r} \ 2^{ln+r} \prod_{\substack{\alpha \in BC(l,1)^r_+ \\ s \in A_\alpha}}(\IP{\alpha}{z}+s\IP{\alpha}{\alpha}). 
    \end{equation*}
\end{prop}

\begin{proof}
    The BA function~\eqref{eqn: psi expanded} must satisfy the condition 2 in Definition~\ref{def:BA}. This gives that
    \begin{equation}\label{conditions}
        \sum_{\nu \in \mathcal{L}_+} c_\nu(z+s\alpha) e^{\IP{\nu +s\alpha}{x}}  = \sum_{\nu \in \mathcal{L}_+} c_\nu(z-s\alpha) e^{\IP{\nu-s\alpha}{x}} 
    \end{equation}
    at $\langle z,\alpha \rangle = 0 $ for $\alpha \in BC(l,1)^r_+$ and $s  \in A_\alpha$. Since $-s\alpha \notin \frac12 \mathcal{L}_+$, while $\nu + s\alpha \in \frac12 \mathcal{L}_+$ for all $\nu \in \mathcal{L}_+$, we get that the term $e^{-\IP{s\alpha}{x}}$ does not appear in the left-hand side of equality~\eqref{conditions}. Hence, it cannot appear in the right-hand side either, which means that $c_0(z-s\alpha) = 0$ at $\langle z,\alpha \rangle = 0$. 
    In other words, the polynomial~$c_0(z)$ must be divisible by
    \begin{equation}\label{factor}
        \prod_{\substack{\alpha \in BC(l,1)_+^r \\ s \in A_\alpha}}(\IP{\alpha}{z}+s\IP{\alpha}{\alpha}).
    \end{equation}
    The quotient of $c_0(z)$ by the product~\eqref{factor} can only be some constant $\lambda$ since the polynomial part of the function~$\psi$ is, by condition 1 in Definition~\ref{def:BA}, of the same degree in $z$ as the polynomial~\eqref{factor}. Moreover, the highest-degree term of the polynomial part of $\psi$ is by definition $\prod_{\alpha \in BC(l,1)_+} \IP{\alpha}{z}^{m_\alpha} = 2^{ln+r}\prod_{\alpha \in BC(l,1)^r_+, \, s \in A_\alpha} \IP{\alpha}{z}$. Thus, denoting by $c_\nu^0$ the highest-degree term of $c_\nu$, equality~\eqref{eqn: psi expanded} implies
    \begin{equation*}
        2^{ln+r}\delta(x)\prod_{\substack{\alpha \in BC(l,1)_+^r \\ s \in A_\alpha}} \IP{\alpha}{z} =
        \lambda e^{-\IP{\rho}{x}}\prod_{\substack{\alpha \in BC(l,1)_+^r \\ s \in A_\alpha}} \IP{\alpha}{z} +
         e^{-\IP{\rho}{x}}\sum_{\nu \in \mathcal{L}_+\setminus\{0\}} c_\nu^0(z) e^{\IP{\nu}{x}}.
    \end{equation*}
    By calculating the coefficient of $e^{-\IP{\rho}{x}}$ in $\delta(x)$, we get that $\lambda = 2^{ln+r}\prod_{\alpha \in BC(l,1)^r_+}(-1)^{m_\alpha + m_{2\alpha}}$, and the statement of the proposition follows.
\end{proof}

We end this section with the following bispectral duality statement, in the spirit of analogous results from~\cite{Chalykh_2000, Feigin_2005} for other configurations. By \cite[Theorem 2.8]{Feigin_2022}, the above BA function~$\psi(z,x)$ is 
a common eigenfunction for a commutative ring of differential operators in $x$, and the following theorem states that a similar situation occurs in the variables $z$.

\begin{thm}\label{thm: commutative rings}
    Let $p(z) \in \mathcal{R}_{BC(l,1)}$ be a polynomial, and let $p_0$ be its highest-degree homogeneous term. Then there exists a differential operator $H_p(x, \partial_x)$ with highest symbol~$p_0(\partial_x)$ and a difference operator~$D_p$ in~$z$ such that
    \begin{equation*}
        H_p \psi(z,x) = p(z)\psi(z,x), \qquad D_p\psi(z,x) = \mu_p(x)\psi(z,x).
    \end{equation*}
    The difference operator is given by
    \begin{equation*}
        D_p = \frac{1}{(\deg p)!} \ad_D^{\deg p} (\widehat{p}),
    \end{equation*}
    where $\widehat{p}$ is the operator of multiplication by $p(z)$ and $\ad^r_A$ is the $r$-th iteration of the operation $\ad_A(B) = AB- BA$. The eigenvalue $\mu_p(x)$ is obtained by substituting $4\sinh(2x_i)$ in place of $z_i$ ($1 \leq i \leq l$) and 
    $4(\sqrt{k})^{-1}\sinh(2\sqrt{k}x_{l+1})$ in place of $z_{l+1}$ into $p_0(z)$.
    
    The operators $H_p$ are pairwise-commuting, and~$H_{-z^2}$ coincides with the Hamiltonian~\eqref{eqn: BC(l,1) CMS Hamiltonian} for $\mathcal{A} = BC(l,1)_+$. 
    The operators $D_p$ commute with $D$ and with each other.
\end{thm}
\begin{proof}
    The differential side of the theorem follows from~\cite[Theorem 2.8, Proposition 2.10]{Feigin_2022}. By Theorem~\ref{thm: bispectrality int case}, we have $D\psi = \mu \psi$. The fact that~$\psi$ is an eigenfunction of the operators $D_p$ follows by a standard argument about bispectral systems (see e.g.\ the proof of \cite[Theorem 4.1]{Chalykh_2000}) which gives that 
    \begin{equation*}
        D_p \psi = \frac{(-1)^{\deg p}}{(\deg p)!} \ad_{\mu}^{\deg p} (H_p) \psi.
    \end{equation*}
    Here $\ad_{\mu}^{\deg p} (H_p) = \ad_{\mu}^{\deg p} (p_0(\partial_{x}))(1)$ is a zeroth-order operator (that is, an ordinary function of $x$), since each application of $\ad_{\mu}$ decreases the order of a linear differential operator and $H_p = p_0(\partial_{x}) + \text{ lower terms}$.
    This also allows, in particular, to compute the eigenvalue.
    
    From the expression for the function $\psi$ given in Proposition~\ref{prop: analogue of 6.5}, it can be seen easily that if a finite difference operator $\widetilde D$ in $z$ (with, say, rational coefficients) is such that $\widetilde D \psi = 0$ identically, then $\widetilde D = 0$. It follows that the operators $D_p$ commute with $D$ and with each other, as $\psi$ is their common eigenfunction.
\end{proof}

\section{Bispectrality for non-integer multiplicities}\label{sec: non-integer multiplicities}
In this section, we carry out an analytic continuation in the parameters of the BA function~\eqref{eqn: BA function} in order to extend to more general complex values of the parameters the statement proved in the previous section about the bispectrality of the 
Sergeev--Veselov operators for $BC(l,1)$. We do this by adapting to our present setting the approach developed in~\cite[Section VI.~C]{Chalykh_2000}. The corresponding generalisation of the function $\psi(z,x)$ will be a deformation of a Heckman--Opdam multidimensional hypergeometric function~\cite{HO}. 

We begin with a few preliminary results on the properties of the generalised CMS operator $H$~\eqref{eqn: BC(l,1) CMS Hamiltonian} for~$\mathcal{A}=BC(l,1)_+$,
where we are now no longer assuming that the parameters $m_\alpha$ are integer.
The next lemma gives a potential-free gauge-equivalent form of this operator. It is stated in~\cite{Sergeev_2004}. We include a proof below for completeness.

\begin{lem}\cite{Sergeev_2004}\label{lem: potential-free gauge}
    Consider the differential operator
    \begin{equation}\label{eqn: L}
        L = \Delta - \sum_{\alpha \in BC(l,1)_+}2m_\alpha \coth\langle \alpha, x \rangle \partial_\alpha,
    \end{equation}
    where $\partial_\alpha = \sum_{i=1}^{l+1} \langle \alpha, e_i \rangle \partial_{x_i}$ is the directional derivative along $\alpha \in \mathbb{C}^{l+1}$. 
    The operator~$L$ is gauge-equivalent to the operator~$H$ from above:
    \begin{equation*}
        L + \rho^2 = - \delta(x) \circ H \circ \delta(x)^{-1},
    \end{equation*}
    where $\rho$ is given by formula~\eqref{eqn: rho} and $\delta(x)$ by formula~\eqref{eqn: delta(x)}.
\end{lem}
\begin{proof}
    We have
\begin{align*}
    \partial_{x_i}[ \delta(x)^{-1}] &=  -\delta(x)^{-1}\sum_{\alpha \in BC(l,1)_+} m_\alpha \alpha^{(i)}\coth\IP{\alpha}{x},
\end{align*}    
\begin{align*}
    \partial_{x_i}^2[ \delta(x)^{-1}] &= \delta(x)^{-1} \sum_{\alpha \in BC(l,1)_+} m_\alpha  \left((m_\alpha + 1) \sinh^{-2}\IP{\alpha}{x} + m_\alpha\right)(\alpha^{(i)})^2 \\
    &\qquad+ \delta(x)^{-1} \sum_{\substack{\alpha,\beta \in BC(l,1)_+ \\ \alpha \neq \beta}} m_\alpha m_\beta  \alpha^{(i)} \beta^{(i)}\coth\IP{\alpha}{x}\coth\IP{\beta}{x}  
\end{align*}
for all $i=1, \dots, l+1$, where $\alpha = (\alpha^{(1)}, \dots, \alpha^{(l+1)})$ and $\beta = (\beta^{(1)}, \dots, \beta^{(l+1)})$. Therefore, 
we get
\begin{equation}
\begin{aligned}\label{eqn: gauge transformed H}
    &-\delta(x) \circ H \circ \delta(x)^{-1} =  
    -\sum_{\alpha \in BC(l,1)_+} 2m_\alpha  m_{2\alpha}\IP{\alpha}{\alpha}\sinh^{-2}\IP{\alpha}{x} +
    \sum_{\alpha \in BC(l,1)_+} m_\alpha^2 \IP{\alpha}{\alpha}  \\
    &\qquad + \sum_{\substack{\alpha,\beta \in BC(l,1)_+ \\ \alpha \neq \beta}} m_\alpha m_\beta  \IP{\alpha}{\beta} \coth\IP{\alpha}{x} \coth\IP{\beta}{x} - \sum_{\alpha \in BC(l,1)_+} 2m_\alpha  \coth\IP{\alpha}{x} \partial_\alpha + \Delta.
\end{aligned}
\end{equation}

By using that $\coth(u)\coth(2u) = \frac12(\sinh^{-2}(u) + 2)$ for any $u \in \mathbb{C} \setminus \frac12 \mathbb{Z}\pi i$, we have 
\begin{align*}
    \sum_{\substack{\alpha,\beta \in BC(l,1)_+ \\ \alpha \neq \beta}} m_\alpha m_\beta  \IP{\alpha}{\beta} \coth\IP{\alpha}{x}\coth\IP{\beta}{x} &=
    \sum_{\substack{\alpha,\beta \in BC(l,1)_+ \\ \alpha \notin \{\beta, 2\beta, \frac12 \beta \}}} m_\alpha m_\beta  \IP{\alpha}{\beta} \coth\IP{\alpha}{x}\coth\IP{\beta}{x} \\
    &\qquad+\sum_{\alpha \in BC(l,1)_+} 2 m_\alpha m_{2\alpha} \IP{\alpha}{\alpha}(\sinh^{-2}\IP{\alpha}{x} + 2);
\end{align*}
and by substituting this into equality~\eqref{eqn: gauge transformed H}, we obtain that $ -\delta(x) \circ H \circ \delta(x)^{-1} $ equals
\begin{align*}
   L+ \sum_{\substack{\alpha,\beta \in BC(l,1)_+ \\ \alpha \notin \{\beta, 2\beta, \frac12 \beta \}}} m_\alpha m_\beta  \IP{\alpha}{\beta} \coth\IP{\alpha}{x}\coth\IP{\beta}{x} + \sum_{\alpha \in BC(l,1)_+} 4m_\alpha m_{2\alpha} \IP{\alpha}{\alpha} + \sum_{\alpha \in BC(l,1)_+} m_\alpha^2 \IP{\alpha}{\alpha}.
\end{align*}
Since we can write
\begin{equation*}
    \rho^2 =  \sum_{\alpha \in BC(l,1)_+} m_\alpha^2 \IP{\alpha}{\alpha} + \sum_{\alpha \in BC(l,1)_+} 4m_\alpha m_{2\alpha} \IP{\alpha}{\alpha} + \sum_{\substack{\alpha,\beta \in BC(l,1)_+ \\ \alpha \notin \{\beta, 2\beta, \frac12 \beta \}}} m_\alpha m_\beta  \IP{\alpha}{\beta},
\end{equation*}
the proof is completed with the help of the following lemma.
\begin{lem}\cite[Equality (12)]{Sergeev_2004}
    We have
    \begin{equation}\label{const}
        \sum_{\substack{\alpha,\beta \in BC(l,1)_+ \\ \alpha \notin \{\beta, 2\beta, \frac12 \beta \}}} m_\alpha m_\beta  \IP{\alpha}{\beta} \coth\IP{\alpha}{x}\coth\IP{\beta}{x} = 
        \sum_{\substack{\alpha,\beta \in BC(l,1)_+ \\ \alpha \notin \{\beta, 2\beta, \frac12 \beta \}}} m_\alpha m_\beta  \IP{\alpha}{\beta}. 
    \end{equation}
\end{lem}
\noindent \textit{\proofname.}
We first indicate how to prove that the left-hand side of~\eqref{const} is non-singular. One can show that singularities at $\sinh \IP{\alpha}{x} = 0$ for $\alpha \in \{e_i, 2e_i \, | \, 1 \leq i \leq l\} \cup \{\sqrt{k}e_{l+1}, 2\sqrt{k}e_{l+1}\} 
\cup \{e_i \pm e_j \, | \, 1 \leq i < j \leq l \}$ cancel by using symmetry --- namely that those $\alpha$ satisfy 
$s_\alpha(BC(l,1)) = BC(l,1)$ --- and using also that for those~$\alpha$ we have $2 \IP{\alpha}{\beta}\IP{\alpha}{\alpha}^{-1} \in \mathbb{Z}$ for all $\beta \in BC(l,1)$.

It remains to show that there are no singularities at $\sinh\langle \alpha, x \rangle=0$ for $\alpha = e_i \pm \sqrt{k}e_{l+1}$ ($1 \leq i \leq l$) either. 
Let $\alpha = e_i + \sqrt{k}e_{l+1}$. Then $\sinh\langle \alpha, x \rangle=0$ if and only if $x_i = i\pi d - \sqrt{k} x_{l+1}$ for some $d \in \mathbb{Z}$. We will show that the terms multiplying $\coth\langle \alpha, x \rangle$ in~\eqref{const} go to $0$ when $x_i \to i\pi d - \sqrt{k}x_{l+1}$. 
The terms in question are (up to a factor of $2$) 
        \begin{equation}
        \begin{aligned}
            m&\coth(x_i)+2n\coth(2x_i) + kp\coth(\sqrt{k}x_{l+1}) + 2kr\coth(2\sqrt{k}x_{l+1}) +(1-k)\coth(x_i - \sqrt{k}x_{l+1}) \\
            &+ \sum_{\substack{j=1 \\ j \neq i}}^l k\bigg(\coth(x_i + x_j) + \coth(x_i - x_j)+\coth(x_j + \sqrt{k}x_{l+1}) - \coth(x_j - \sqrt{k}x_{l+1}) \bigg). \label{coths}
        \end{aligned}
        \end{equation}
        By using $i\pi$-periodicity of $\coth$, the expression~\eqref{coths} simplifies at $x_i=i\pi d - \sqrt{k}x_{l+1}$ to
        \begin{equation*}
           (kp-m)\coth(\sqrt{k}x_{l+1}) +(2kr-2n + k-1)\coth(2\sqrt{k}x_{l+1}) = 0,
        \end{equation*}
       as required, since $m = kp$ and $2n+1 = k(2r+1)$. The case of $\alpha = e_i-\sqrt{k}e_{l+1}$ can be handled similarly.

This completes the proof that the left-hand side of~\eqref{const} is non-singular; and from its form, it is then clear that it can be rewritten in exponential variables $w_i = e^{2x_i}$ ($1\leq i \leq l$) and $w_{l+1} = e^{2\sqrt{k}x_{l+1}}$ as a non-singular rational function with zero degree, which therefore must be constant. 

To calculate the constant
, let us put $x = ((l+1)N, lN, \dots, 2N, N/\sqrt{k})$, and take the limit $N \to \infty$ by using that $\IP{\alpha}{x} \to \infty$ as $N \to \infty$ for all $\alpha \in BC(l,1)_+$ and that $\coth u \to 1$ when $u \to \infty$. This gives the right-hand side of equality~\eqref{const}.
\qedhere\qedsymbol
\end{proof}

When $x$ belongs to the region 
\begin{equation*}
    B = B(k) = \{ x \in \mathbb{C}^{l+1} \, | \, \re\IP{\alpha}{x} < 0 \text{ for all } \alpha \in BC(l,1)_+\},
\end{equation*}
the operator~\eqref{eqn: L} can be expanded into a series as
\begin{equation}\label{eqn: alt form of L}
    L = \Delta + \sum_{\alpha \in BC(l,1)_+}2m_\alpha \frac{1+e^{2\langle \alpha, x \rangle}}{1-e^{2\langle \alpha, x \rangle}} \partial_\alpha = \Delta + \sum_{\alpha \in BC(l,1)_+}2m_\alpha \left(1 + 2\sum_{j=1}^\infty e^{2j\langle \alpha, x \rangle} \right)\partial_\alpha.
\end{equation}
Let $\varphi = \varphi(z,x)$ be a solution of the equation
\begin{equation}\label{eqn: SE for L}
    L \varphi = (z^2 - \rho^2) \varphi,
\end{equation}
which is by Lemma~\ref{lem: potential-free gauge} equivalent to the function $\delta(x)^{-1}\varphi$ being an eigenfunction of the Hamiltonian $H$ with eigenvalue $-z^2$. In particular, if $m_\alpha \in \mathbb{Z}_{\geq 0}$ then $\delta(x) \psi(z,x)$, where $\psi$ is the Baker--Akhiezer function~\eqref{eqn: BA function}, satisfies equation~\eqref{eqn: SE for L}. 

Assume that there is a solution $\varphi$ of the equation~\eqref{eqn: SE for L} of the particular form
\begin{equation}\label{eqn: phi expanded}
    \varphi = e^{\IP{z-\rho}{x}} \sum_{\nu \in \mathcal{L}_+} q_\nu(z) e^{\IP{\nu}{x}}
\end{equation}
for some functions $q_\nu(z)$ with $q_0(z) = 1$. If $m_\alpha \in \mathbb{Z}_{\geq 0}$ then by Proposition~\ref{prop: analogue of 6.5} the function $\delta(x)\psi(z,x)$ is of the form~\eqref{eqn: phi expanded}; it is just normalised differently since 
$c_0(z) = 2^{ln + r}\prod_{\alpha \in BC(l,1)^r_+, s \in A_\alpha} (- \IP{\alpha}{z} - s \IP{\alpha}{\alpha}) \neq 1$ by Proposition~\ref{prop: analogue of 6.6}.

By substituting the series~\eqref{eqn: phi expanded} into equation~\eqref{eqn: SE for L}, using the expansion~\eqref{eqn: alt form of L} for the operator $L$, and requiring that the respective coefficients of the terms $e^{\IP{\nu}{x}}$ vanish for all $\nu \in \mathcal{L}_+$, we get recurrence equations for the functions $q_\nu$. Namely, we get
\begin{equation}\label{eqn: recurrence relations}
    \IP{\nu}{\nu + 2z}q_\nu(z) + \sum_{\alpha \in BC(l,1)_+} 4m_\alpha \sum_{j=1}^{\infty} \IP{\alpha}{z-\rho + \nu - 2j\alpha} q_{\nu-2j\alpha}(z) = 0,
\end{equation}
subject to the constraint that $q_0 = 1$, and where we put $q_{\nu - 2j\alpha} = 0$ if $\nu - 2j\alpha \notin \mathcal{L}_+$, which means that the above sum over $j$ is finite.
For generic~$z$, equations~\eqref{eqn: recurrence relations} (together with the normalisation choice $q_0=1$) determine all~$q_\nu$ uniquely. Indeed, one can solve for them recursively by height of $\nu \in \mathcal{L}_+ =\oplus_{i=1}^{l+1} \mathbb{Z}_{\geq 0} \alpha_i$, since $\nu - 2j\alpha$ has a strictly lower height than $\nu$, where for $\nu = \sum_{i=1}^{l+1} \nu^{(i)} \alpha_i$, its height is defined by $h(\nu) = \sum_{i=1}^{l+1} \nu^{(i)}$. We note that all~$q_\nu(z)$ are rational functions of~$z$, and their dependence on $m,n,p,r$, and~$\sqrt{k}$ is also rational. 

For $k$ with $\re k > 0$, the next lemma below, applied with $x$ replaced by $x/2$, can be used to show that the resulting series~\eqref{eqn: phi expanded} converges absolutely in the region~$B$. 
Moreover, since it is a power series in $y_i = e^{\IP{\alpha_i}{x}}$ ($i=1, \dots, l+1$), it converges uniformly in the open sets $\{ x \in B \, | \, \re\IP{\alpha_i}{x} < \varepsilon < 0,$ $i=1, \dots, l+1 \}$ for any~$\varepsilon < 0$, and 
so $\varphi$ is analytic in $x$ on $B$. In the case of root systems, an analogue of the next lemma is proved in~\cite[Lemma 5.3]{Helgason}.

\begin{lem}\label{lem: Helgason}
    Assume $\re k > 0$.
    For $\nu \in \mathcal{L}_+ \setminus \{0\}$, let $\Sigma_\nu 
    = \{ z \in \mathbb{C}^{l+1} \, | \, \IP{\nu}{\nu+2z} = 0 \}$. 
    Suppose $z \in \mathbb{C}^{l+1}$ does not lie on any of the hyperplanes $\Sigma_\nu$, and let $x \in B$. Then there exists a constant $K=K(z,x) \in \mathbb{R}$ (depending on $z$ and $x$ but not on $\nu$) such that 
    \begin{equation*}
        |q_\nu(z) e^{\IP{\nu}{x}}| \leq K 
    \end{equation*}
    for all $\nu \in \mathcal{L}_+$.
\end{lem}

\begin{proof}
    Let $\nu = \sum_{i=1}^{l+1} \nu^{(i)} \alpha_i \in \mathcal{L}_+$ and $\alpha \in BC(l,1)_+$ be arbitrary. We have
    \begin{equation}\label{inequality 1}
        |\IP{\nu-\rho+z}{\alpha}| \leq |\IP{z-\rho}{\alpha}| + \sum_{i=1}^{l+1} \nu^{(i)} |\IP{\alpha_i}{\alpha}| \leq \lambda_1 (h(\nu) + 1),
    \end{equation}
    where we let $\lambda_1 = \lambda_1(z) > 0$ be the maximum (depending on $z$ but not on $\nu$) of the finite set
    $$
    \{|\IP{z-\rho}{\beta}| \colon \beta \in BC(l,1)_+\} \cup
    \{|\IP{\alpha_i}{\beta}| \colon \beta \in BC(l,1)_+, \ i=1, \dots, l+1\}.
    $$ 
     
    Further, we have 
    \begin{equation}\label{nu^2}
        \re \IP{\nu}{\nu} = \sum_{i,j = 1}^{l+1} \nu^{(i)} \nu^{(j)} \IP{\widetilde \alpha_i}{\widetilde \alpha_j} = 
        \lVert \widetilde \nu \rVert^2
    \end{equation}
    for $\widetilde \nu = \sum_{i=1}^{l+1} \nu^{(i)} \widetilde \alpha_i$, where $\widetilde \alpha_i \in \mathbb{R}^{l+1}$ are obtained from $\alpha_i$ by replacing $\sqrt{k}$ with $\sqrt{\re k} \in \mathbb{R}$. 
    Here $\lVert \cdot \rVert$ denotes the usual (real) Euclidean norm.
    Since $\widetilde \alpha_i$ form a basis of $\mathbb{R}^{l+1}$,
    the expression~\eqref{nu^2} is a positive-definite (real) quadratic form in the variables $\nu^{(i)}$ with associated symmetric matrix 
    $A = (\IP{\widetilde \alpha_i}{\widetilde \alpha_j})_{i,j = 1}^{l+1}$. By Sylvester's criterion, all leading principal minors~$M_i$ of~$A$ are positive, $M_i = \det A_i$, where $A_i$ is the top left~$i\times i$ corner of $A$. For any $c \in \mathbb{R}$, the expression
    $
        \re \IP{\nu}{\nu} - c h(\nu)^2
    $
    is also a quadratic form in the variables~$\nu^{(i)}$ with associated matrix $A^c = (\IP{\widetilde \alpha_i}{\widetilde \alpha_j} - c)_{i,j = 1}^{l+1}$. By the matrix determinant lemma, the leading principal minors $M_i^c$ of $A^c$ are 
    $(1 - c \IP{u}{A_i^{-1}u})M_i$, where $u = (1, \dots, 1) \in \mathbb{R}^i$.
    Since $M_i$ are positive, $M_i^c$ are, too, for $c = \min \{ 2^{-1}\IP{u}{A_i^{-1}u}^{-1} \, | \, i=1, \dots, l+1 \} > 0$, and then Sylvester's criterion implies
    $$ |\IP{\nu}{\nu}| \geq \re \IP{\nu}{\nu} \geq c h(\nu)^2.$$

    We also have  
    \begin{equation*}
        |\IP{\nu}{z}| \leq \sum_{i=1}^{l+1} \nu^{(i)} |\IP{\alpha_i}{z}| \leq M h(\nu),
    \end{equation*}
    where we let $M = M(z) = \max \{|\IP{\alpha_i}{z}| \colon i=1, \dots, l+1\} > 0$.
    Whenever $h(\nu) \geq 4 M/c$, it follows that $|\IP{\nu}{z}| \leq ch(\nu)^2/4$, and then by the reverse triangle inequality
    \begin{equation*}
        |\IP{\nu}{\nu+2z}| \geq \big| |\IP{\nu}{\nu}| - 2|\IP{\nu}{z}|  \big| \geq c h(\nu)^2/2.
    \end{equation*}
    Letting $\lambda_2 = \lambda_2(z) > 0$ be the minimum of $c/2$ and the finitely-many, positive values $|\IP{\nu'}{\nu'+2z}|/h(\nu')^2$ for $\nu' \in \mathcal{L}_+ \setminus \{0\}$ with $h(\nu') < 4 M/c$ (here we are using that $z \notin \Sigma_{\nu'}$ by assumption), we thus get that
    \begin{equation}\label{inequality 2}
        |\IP{\nu}{\nu+2z}| \geq \lambda_2 h(\nu)^2
    \end{equation}
    for any $\nu \in \mathcal{L}_+$.

    By using inequalities~\eqref{inequality 1}, \eqref{inequality 2}, and the recurrence relation~\eqref{eqn: recurrence relations}, we get for all $\nu \in \mathcal{L}_+ \setminus \{0\}$ that
    \begin{equation}\label{first estimate}
        |q_\nu(z)| \leq 4 \lambda_1 \lambda_2^{-1} h(\nu)^{-1} \sum_{\alpha \in BC(l,1)_+} |m_\alpha| \sum_{j=1}^{\infty} |q_{\nu - 2j\alpha}(z)|,
    \end{equation}
    since $h(\nu - 2j\alpha) \leq h(\nu)-1$. Let $\lambda = 4 \lambda_1 \lambda_2^{-1}$.

    Since $x \in B$ and the geometric series is absolutely convergent on the open unit disk in~$\mathbb{C}$, there is~$N_0 \in \mathbb{Z}_{>0}$ such that
    \begin{equation}\label{N_0}
        \lambda \sum_{\alpha \in BC(l,1)_+} |m_\alpha| \sum_{j=1}^{\infty} |e^{2j \IP{\alpha}{x}}| \leq N_0.
    \end{equation}
    Let $K$ be such that
    \begin{equation}\label{base case}
        |q_\eta(z) e^{\IP{\eta}{x}}| \leq K
    \end{equation}
    for those (finitely many) $\eta \in \mathcal{L}_+$ which have $h(\eta) \leq N_0$. One can prove that~\eqref{base case} holds for all $\eta \in \mathcal{L}_+$ by induction on $h(\eta)$. Indeed, assume that~\eqref{base case} holds for those $\eta$ with $h(\eta) < N$ for some 
    integer $N > N_0$. Then for~$\nu \in \mathcal{L}_+$ with $h(\nu) = N$, we have by inequalities~\eqref{first estimate}, \eqref{N_0}, and the induction hypothesis that
    \begin{equation*}
        |q_\nu(z)| \leq \lambda N^{-1} \sum_{\alpha \in BC(l,1)_+} |m_\alpha| \sum_{j=1}^{\infty} K |e^{-\IP{\nu - 2j\alpha}{x}}|
        \leq K N_0 N^{-1}|e^{-\IP{\nu}{x}}| < K|e^{-\IP{\nu}{x}}|,  
    \end{equation*}
    which completes the proof by induction.
\end{proof}

As a corollary of Lemma~\ref{lem: Helgason}, we have the following statement. 
\begin{prop}\label{rem: analyticity}
    The series~\eqref{eqn: phi expanded} defines an analytic function in $z$, $x$, and $k$ on an open subset of~$\mathbb{C}^{2l+3}$.
\end{prop}
\begin{proof}
Let $k_0 \in \mathbb{C}$ with $\re k_0 > 0$. 
Since $\re \IP{\alpha_i}{x}$ is continuous in $x$ and $k$, there exists an open ball~$\mathcal{B}_{\mathfrak K}$ centred at $k_0$ and 
an open ball $\mathcal{B}_{\mathfrak X} \subset \mathbb{C}^{l+1}$ such that $\re k > 0$ and $\mathcal{B}_{\mathfrak X} \subset B(k)$  for all $k \in \mathcal{B}_{\mathfrak K}$. Take any~$z_0 \in \mathbb{C}^{l+1} \setminus \cup_{\nu \in \mathcal{L}(k_0)_+ \setminus \{0\}} \Sigma_\nu$.
Consider the constant $c = c(k_0)$ from the proof of Lemma~\ref{lem: Helgason} for~$k = k_0$. Note that $\IP{u}{A_i^{-1}u}^{-1}$ is continuous in $k$ at $k_0$. Consider also the constant $M = M(z_0, k_0)$ defined as in the proof of Lemma~\ref{lem: Helgason}. Note that~$|\IP{\alpha_i}{z}|$ is continuous in $z$ and $k$ at~$(z_0, k_0)$. Therefore, there exists an open ball $\mathcal{B}_{\mathfrak K}' \subseteq \mathcal{B}_{\mathfrak K}$ centred at $k_0$ and an open ball $\mathcal{B}_{\mathfrak Z} \subset \mathbb{C}^{l+1}$ centred at $z_0$ such that
for all $(z,k) \in \mathcal{B}_{\mathfrak Z} \times \mathcal{B}_{\mathfrak K}'$ we have $c(k) > c(k_0)/2$ and $M(z,k) < 2 M(z_0, k_0)$.

Since $|\IP{\nu(k_0)}{\nu(k_0) + 2z_0}| > 0$ for all $\nu(k_0) \in \mathcal{L}(k_0)_+$, and $|\IP{\nu(k)}{\nu(k) + 2z}|$ is continuous in $k$ and $z$, there exists an open ball $\mathcal{B}_{\mathfrak K}'' \subseteq \mathcal{B}_{\mathfrak K}'$ and an open ball $\mathcal{B}_{\mathfrak Z}' \subseteq \mathcal{B}_{\mathfrak Z}$ such that for the \emph{finitely many} $\nu = \nu(k) \in \mathcal{L}(k)_+$ with $h(\nu) \leq 16 M(z_0, k_0)/c(k_0)$ we have 
$|\IP{\nu}{\nu+2z}| > 0$  
for all $k \in \mathcal{B}_{\mathfrak K}''$ and $z \in \mathcal{B}_{\mathfrak Z}'$. On the other hand, for those $\nu$ with $h(\nu) > 16 M(z_0, k_0)/c(k_0)$, we have by the definition of the constants~$c(k)$ and~$M(z,k)$ that $|\IP{\nu}{\nu}| \geq c(k)h(\nu)^2 > c(k_0)h(\nu)^2/2$ and $|\IP{\nu}{z}| \leq M(z,k) h(\nu) < 2 M(z_0,k_0) h(\nu) < c(k_0)h(\nu)^2/8$, and so by the reverse triangle inequality 
\begin{equation*}
    |\IP{\nu}{\nu+2z}| \geq \big|  |\IP{\nu}{\nu}| - 2|\IP{\nu}{z}| \big| > c(k_0)h(\nu)^2/4 > 0. 
\end{equation*}

In other words, for any $(z,x,k) \in U \coloneqq \mathcal{B}_{\mathfrak Z}' \times \mathcal{B}_{\mathfrak X} \times \mathcal{B}_{\mathfrak K}''$, we have $\re k > 0$, $z \in \mathbb{C}^{l+1} \setminus \cup_{\nu \in \mathcal{L}(k)_+ \setminus \{0\}} \Sigma_\nu$, and $x \in B(k)$. 
Then the sum~$\varphi(z,x)$ of the series~\eqref{eqn: phi expanded} is well-defined by the discussion preceding Lemma~\ref{lem: Helgason}, and~$\varphi(z,x)$ is on $U$ the pointwise limit of a sequence of functions (the partial sums) that are analytic \emph{jointly} in all the variables~$z$, $x$, and $k$. 
As a consequence of (the multivariable version of) Osgood's Theorem~\cite[Theorem 8]{Krantz}, there exists an open dense subset~$V \subseteq U$ on which~$\varphi(z,x)$ is analytic (jointly) in the variables~$z$,~$x$, and $k$ and on which the convergence is locally uniform (for the original, single-variable case of Osgood's Theorem see~\cite[Theorem II]{Osgood}). 
\end{proof}


If $m_\alpha \in \mathbb{Z}_{\geq 0}$, then by the uniqueness of the solution of the system~\eqref{eqn: recurrence relations}, we must have that~$\delta(x)\psi(z,x)$  
is proportional to $\varphi(z,x)$ with the factor of proportionality being $c_0(z)$, that is 
\begin{equation*}
    \psi(z,x) = 2^{ln + r} \delta(x)^{-1} \varphi(z,x)\prod_{\substack{\alpha \in BC(l,1)^r_+ \\ s \in A_\alpha}} (- \IP{\alpha}{z} - s \IP{\alpha}{\alpha}).
\end{equation*}

For the case when $m_\alpha$ are not necessarily in $\mathbb{Z}_{\geq 0}$, but rather we have arbitrary $k, m,n,p,r \in \mathbb{C}$ with $\re k > 0$, $m=kp$, and $2n+1 = k(2r+1)$, let us define the following function 
\begin{equation*}
    \Psi(z,x) = C(z) \delta(x)^{-1} \varphi(z,x),
\end{equation*}
where we take any branch of $\delta(x)$, and where we defined the function 
\begin{align*}
    C(z) &= \prod_{\alpha \in \{\sqrt{k}e_{l+1}, \, e_i \, | \, 1 \leq i \leq l\}} 
    \frac{\Gamma\left(-\IP{\alpha}{z}\IP{\alpha}{\alpha}^{-1}\right)\Gamma\left(-\frac12\IP{\alpha}{z}\IP{\alpha}{\alpha}^{-1} - \frac12 m_\alpha \right)}{\Gamma\left(-\IP{\alpha}{z}\IP{\alpha}{\alpha}^{-1} - m_\alpha \right)\Gamma\left(-\frac12\IP{\alpha}{z}\IP{\alpha}{\alpha}^{-1} - \frac12m_\alpha - m_{2\alpha}\right)}  \\
    &\quad \times \prod_{\alpha \in \{e_i \pm e_j \, | \, 1 \leq i < j \leq l\}} 
    \frac{\Gamma\left(-\IP{\alpha}{z}\IP{\alpha}{\alpha}^{-1}\right)}{\Gamma\left(-\IP{\alpha}{z}\IP{\alpha}{\alpha}^{-1} - k \right)} \prod_{\alpha \in \{ e_i \pm \sqrt{k}e_{l+1} \,|\, 1 \leq i \leq l \}} (-\IP{\alpha}{z} - 1 - k).
\end{align*}
Here $\Gamma(u)$ is the classical gamma-function.
Then $H \Psi = -z^2 \Psi$, since $C(z)$ does not depend on~$x$ and~$\varphi$ solves equation~\eqref{eqn: SE for L} by construction.
Moreover, $\Psi(z,x)$ coincides (up to a constant factor) with $\psi(z,x)$ when $m_\alpha \in \mathbb{Z}_{\geq 0}$, since then~$C(z) \propto c_0(z)$ as $\Gamma(u)/\Gamma(u-N) = \prod_{i=1}^N(u-i)$ for $u \in \mathbb{C}$, $N \in \mathbb{Z}_{\geq 0}$. 

With that, the proof of the next theorem is then essentially the same as that of~\cite[Theorem 6.9]{Chalykh_2000}. It just uses Theorem~\ref{thm: bispectrality int case} in place of~\cite[Theorem 6.2]{Chalykh_2000}.  

\begin{thm}\label{thm: bispectrality non-int case}
For any $k, m,n,p,r \in \mathbb{C}$ with $\re k>0$, $m=kp$, and $2n+1 = k(2r+1)$, the function~$\Psi(z,x)$ satisfies 
\begin{align}
    H \Psi &= -z^2 \Psi, \nonumber \\
    D \Psi &= \mu(x)\Psi. \label{eqn: D Psi}
\end{align}
\end{thm}
\begin{proof}
    It only remains to show equality~\eqref{eqn: D Psi}.
    For $m_\alpha \in \mathbb{Z}_{\geq 0}$, it follows from Theorem~\ref{thm: bispectrality int case}.
    More generally,
    equation~\eqref{eqn: D Psi} is equivalent to $\varphi$ satisfying $\widetilde D \varphi = \mu(x) \varphi$ for the difference operator 
    \begin{equation*}
        \widetilde D = C(z)^{-1} \circ D \circ C(z) = \sum_{\tau \in S} a_\tau(z) \left(C(z)^{-1} C(z+\tau) T_\tau - 1\right),
    \end{equation*}
    where $S=\{\pm 2e_i, \pm 2\sqrt{k}e_{l+1} \,|\, 1 \leq i \leq l\}$, since $D$ was of the form $D = \sum_{\tau \in S} a_\tau(z) (T_\tau - 1)$.
    For any $\tau \in S$, the function $a_\tau(z)$ is rational in $z,m,n,p,r, \sqrt{k}$, and so is the function $C(z)^{-1} C(z+\tau)$, since $\frac12\IP{\alpha}{\tau}\IP{\alpha}{\alpha}^{-1} \in \mathbb{Z}$ for $\alpha \in \{\sqrt{k}e_{l+1}, \, e_i \, | \, 1 \leq i \leq l\}$ and $\IP{\alpha}{\tau}\IP{\alpha}{\alpha}^{-1} \in \mathbb{Z}$ for $\alpha \in \{e_i \pm e_j \, | \, 1 \leq i < j \leq l\}$.
    
    By substituting the series~\eqref{eqn: phi expanded} into the equation $\widetilde D \varphi = \mu(x) \varphi$, the latter reduces (by looking at the coefficient of~$e^{\IP{\nu}{x}}$ for each $\nu \in \mathcal{L} \supset S$) to an infinite number of identities, each involving a finite number of the coefficients $q_\nu(z)$ and only involving rational functions of $z$, $m,n,p,r, \sqrt{k}$. Explicitly, these identities are
    \begin{equation*}
        \sum_{\tau \in S} (\kappa_\tau - a_\tau(z)) q_{\nu}(z) + \sum_{(\mu, \tau) \in \mathcal{L}_+ \times S \colon \mu + \tau = \nu}
        \left(a_{\tau}(z)C(z)^{-1}C(z+\tau) q_\mu(z + \tau) - \kappa_\tau q_\mu(z) \right) = 0
    \end{equation*}
    for $\nu \in \mathcal{L}$, where we put $q_\nu = 0$ if $\nu \notin \mathcal{L}_+$.
    Theorem~\ref{thm: bispectrality int case} implies that these identities hold when $m_\alpha \in \mathbb{Z}_{\geq 0}$, and then it follows that they hold in general. This completes the proof.
\end{proof}

By analyticity, the bispectrality property of Theorem~\ref{thm: bispectrality non-int case} holds in a bigger domain of analyticity of the function $\varphi$. To be more precise, 
by Proposition~\ref{rem: analyticity} we have on an open set $V$ an analytic function $\varphi$ that satisfies equation~\eqref{eqn: SE for L} and by the proof of Theorem~\ref{thm: bispectrality non-int case} also $\widetilde
     D \varphi = \mu(x) \varphi$. Suppose $(z,x,k) \in \mathbb{C}^{l+1} \times \mathbb{C}^{l+1} \times \mathbb{C}$ is such that~$\varphi$ can be analytically extended to a function~$\widetilde \varphi(z,x,k)$ on some neighbourhood $W$ of $(z,x,k)$ containing $V$. 
    The function $L \widetilde \varphi - (z^2 - \rho^2)\widetilde \varphi$ is analytic in $z,x,k$ away from the singularities of $\coth \IP{\alpha}{x}$ for $\alpha \in BC(l,1)_+$, and on $V$ it is identically zero. Thus, it must be zero on all of its domain of analyticity.
    Similarly, $\widetilde D \widetilde \varphi - \mu(x) \widetilde \varphi$ is analytic in $z,x,k$ away from the union $P$ of the poles of the functions $a_\tau(z)$ and $C(z+\tau)C(z)^{-1}$, and it vanishes on the open set~$V \setminus P$. Hence $\widetilde D \widetilde \varphi = \mu(x) \widetilde \varphi$ on $W \setminus P$.
    In terms of the function $\widetilde \Psi \coloneqq C \widetilde \varphi/ \delta$, this means that the following bispectrality relation is satisfied:
    \begin{align*}
        H \widetilde \Psi &= -z^2 \widetilde \Psi, \\
        D \widetilde \Psi &= \mu(x)\widetilde \Psi. 
    \end{align*}

\subsection*{Acknowledgements}
We are very grateful to our adviser Misha Feigin for suggesting to us this research problem and for his invaluable guidance and help, especially with Proposition~\ref{prop: analogue of 6.5} and Section~\ref{sec: non-integer multiplicities}. M.\,V.\ would also like to thank Oleg Chalykh for helpful discussions.

The work of I.\,M.\ was supported by a Student Research Bursary in Mathematics and Statistics from the Edinburgh Mathematical Society. L.\,R.\ was supported by an Undergraduate Research Bursary from the London Mathematical Society. 
M.\,V.\ was funded by a Carnegie--Caledonian PhD scholarship from the Carnegie Trust for the Universities of Scotland.

{\small
\subsection*{Data availability}
Data sharing is not applicable to this article as no datasets were generated or analysed.
}

\bibliographystyle{plain}\bibliography{ref.bib}
\end{document}